%% file: main.tex
\documentclass[a4paper,11pt]{article}
\usepackage[T1]{fontenc}
\usepackage[top=1in,left=1in,right=1in,bottom=1in]{geometry}
\usepackage{amsmath,amsthm,amssymb}
\usepackage{mathtools}
\usepackage{color}
\usepackage{graphicx}
\usepackage{enumitem}
\usepackage{xspace}
\usepackage{placeins}
\usepackage{authblk}
\usepackage{wrapfig}
\usepackage[font=scriptsize,labelfont=bf,skip=0.5em]{caption}
\usepackage[hidelinks]{hyperref}
\setitemize{noitemsep,topsep=0pt,parsep=0pt,partopsep=0pt}

\newcommand{\newtext}[1]{#1}
\newcommand{\oldtext}[1]{}

\newif\ifMdBtheorems
\MdBtheoremstrue
\ifMdBtheorems
    \newtheorem{defin}{Definition}
        \newenvironment{definition}{\begin{defin} \sl}{\end{defin}}
    \newtheorem{theo}[defin]{Theorem}
        \newenvironment{theorem}{\begin{theo} \sl}{\end{theo}}
    \newtheorem{lem}[defin]{Lemma}
        \newenvironment{lemma}{\begin{lem} \sl}{\end{lem}}
        
    \newtheorem{propo}[defin]{Proposition}
        \newenvironment{proposition}{\begin{propo} \sl}{\end{propo}}
    \newtheorem{coro}[defin]{Corollary}
        \newenvironment{corollary}{\begin{coro} \sl}{\end{coro}}
    \newtheorem{obse}[defin]{Observation}
        \newenvironment{observation}{\begin{obse} \sl}{\end{obse}}
    \newtheorem{rem}[defin]{Remark}
        \newenvironment{remark}{\begin{rem} \rm}{\end{rem}}
    \newtheorem{myfact}[defin]{Fact}
        
\else
    \theoremstyle{plain}
    \newtheorem{theorem}{Theorem}
    \newtheorem{lemma}[theorem]{Lemma}
    
    \newtheorem{corollary}[theorem]{Corollary}
    \newtheorem{observation}[theorem]{Observation}

    \theoremstyle{definition}
    \newtheorem{definition}[theorem]{Definition}

    \theoremstyle{remark}
    \newtheorem{remark}[theorem]{Remark}
\fi

\definecolor{ablue}{rgb}{0.3,0.4,0.8}
\definecolor{ared}{rgb}{0.95,0.4,0.4}
\definecolor{agreen}{rgb}{0,0.5,0.25}
\definecolor{ayellow}{rgb}{0.95,0.85,0.3}

\DeclarePairedDelimiter{\civ}{[}{]}

\newcommand{\ceil}[1]{\left\lceil #1 \right\rceil}

\newcommand{\NP}{\mbox{\ensuremath{\mathsf{NP}}}\xspace}

\newcommand{\HC}{\textsc{Hamiltonian Cycle}\xspace}

\newcommand{\IS}{\textsc{Independent Set}\xspace}

\newcommand{\etsp}{{\sc Euclidean TSP}\xspace}

\newcommand{\Reals}{\mathbb{R}}
\newcommand{\Nats}{\mathbb{N}}
\newcommand{\Ints}{\mathbb{Z}}

\newcommand{\cC}{\mathcal{C}}
\newcommand{\cP}{\mathcal{P}}

\newcommand{\cD}{\mathcal{D}}

\newcommand{\cU}{\mathcal{U}}

\newcommand{\bn}{\mathbf{n}}
\newcommand{\bx}{\mathbf{x}}
\newcommand{\by}{\mathbf{y}}

\newcommand{\bp}{\mathbf{p}}
\newcommand{\bq}{\mathbf{q}}

\newcommand{\sig}{\sigma}
\newcommand{\eps}{\varepsilon}

\newcommand{\mybox}{\mathsf{Box}}
\newcommand\eqdef{\mathrel{\overset{\makebox[0pt]{\mbox{\normalfont\tiny\sffamily def}}}{=}}}

\DeclareMathOperator{\diam}{diam}
\DeclareMathOperator{\radius}{radius}
\newcommand{\myin}{\mathrm{in}}
\newcommand{\myout}{\mathrm{out}}
\newcommand{\sep}{\mathrm{sep}}
\newcommand{\mylarge}{\mathrm{large}}
\newcommand{\Flarge}{F_\mylarge}
\newcommand{\sub}{\mathrm{sub}}
\newcommand{\Bin}{B_{\myin}}
\newcommand{\Bout}{B_{\myout}}
\newcommand{\bags}{\mathrm{Bags}}
\DeclareMathOperator{\weight}{weight}

\newcommand{\iG}[1]{G[#1]}
\newcommand{\union}{U}

\newcommand{\bd}{\partial}

\DeclareMathOperator{\poly}{poly}
\DeclareMathOperator{\emb}{Emb}

\renewcommand{\leq}{\leqslant}
\renewcommand{\geq}{\geqslant}

\newcommand{\etal}{\emph{et~al.}}

\newcommand{\BeginMyItemize}{\begin{itemize}\setlength{\itemsep}{-\parskip}}
\newcommand{\EndMyItemize}{\end{itemize}}

\newcommand{\BeginMyEnumerate}{\begin{enumerate}\setlength{\itemsep}{-\parskip}}
\newcommand{\EndMyEnumerate}{\end{enumerate}}

\newcommand{\defproblem}[3]{
\vspace{1mm}
\noindent
\fbox{
  \begin{minipage}{0.96\textwidth}
  \begin{tabular*}{\textwidth}{@{\extracolsep{\fill}}lr}
    \textsc{#1} &  \\
  \end{tabular*}
  {\bf{Input:}} #2 \\
  {\bf{Question:}} #3
  \end{minipage}
}
\vspace{1mm}
}

\newcommand{\IndS}{{\sc Independent Set}\xspace}

\author[1]{Mark de Berg}
\author[2]{Hans L. Bodlaender}
\author[3]{S\'andor Kisfaludi-Bak}
\author[4]{\\D\'aniel Marx}
\author[5]{Tom C. van der Zanden}

\affil[1]{Eindhoven University of Technology, The Netherlands\authorcr\texttt{M.T.d.Berg@tue.nl}}
\affil[2]{Utrecht University, The Netherlands\authorcr\texttt{H.L.Bodlaender,@uu.nl}}
\affil[3]{Max Planck Institute for Informatics, Saarbr\"ucken, Germany \authorcr\texttt{sandor.kisfaludi-bak@mpi-inf.mpg.de}}
\affil[4]{CISPA Helmholtz Center for Information Security, Saarbr\"ucken, Germany \authorcr\texttt{marx@cispa.saarland}}
\affil[5]{Maastricht University, The Netherlands \authorcr\texttt{T.vanderZanden@maastrichtuniversity.nl}}

\title{A Framework for ETH-Tight Algorithms and Lower\\ Bounds in Geometric Intersection Graphs\thanks{This work was supported by the NETWORKS project, funded by the Netherlands Organization for Scientific Research NWO under project no. 024.002.003. and by the
ERC Consolidator Grant SYSTEMATICGRAPH (No.~725978) of the European Research
Council.}}

\date{}

\begin{document}

\maketitle

\begin{abstract}
We give an algorithmic and lower-bound framework that facilitates the construction of subexponential algorithms and matching conditional complexity bounds. It can be applied to  
\oldtext{OLD: a wide range of geometric intersection graphs (intersections of similarly sized fat objects),}
\newtext{intersection graphs of similarly-sized fat objects,}
yielding algorithms with running time $2^{O(n^{1-1/d})}$ for any fixed dimension $d\ge 2$ for many well-known graph problems, including \textsc{Independent Set}, $r$-\textsc{Dominating Set} for constant $r$, and \textsc{Steiner Tree}. For most problems, we get improved running times compared to prior work; in some cases, we give the first known subexponential algorithm in geometric intersection graphs. Additionally, most of the obtained algorithms are representation-agnostic, i.e., they work on the graph itself and do not require the geometric representation. Our algorithmic framework is based on a weighted separator theorem and various treewidth techniques.

The lower bound framework is based on a constructive embedding of graphs into $d$-dimensional grids, and it allows us to derive matching $2^{\Omega(n^{1-1/d})}$ lower bounds under the Exponential Time Hypothesis even in the much more restricted class of $d$-dimensional induced grid graphs.
\end{abstract}

\input{intro.tex}

\input{algorithmic_framework.tex}

\input{lower_framework.tex}

\input{conclusion.tex}

\bibliography{highdim-h}

\clearpage

\appendix

\section{Problem definitions}\label{app:probdefs}
 
In this section, we state the formal definitions of problems appearing in our paper.

\defproblem{$(3,3)$-SAT}{A CNF formula $\phi$ with at most 3 variables per clause and where each variable occurs in at most 3 clauses.}{Is there is a satisfying assignment?}

\defproblem{Grid Embedded SAT}{A $(3,3)$-CNF formula $\phi$ whose incidence graph $G_{\phi}$ is embeddable in $G_2(n,n)$.}{Is there is a satisfying assignment?}

\defproblem{Planar SAT}{A CNF formula $\phi$ whose incidence graph $G_{\phi}$ is planar.}{Is there is a satisfying assignment?}

\defproblem{Independent Set}
{A graph $G=(V,E)$ and an integer $k$}
{Decide if there is a vertex set $I\subseteq V$ of size $k$ that induces no edges.}

\defproblem{Vertex Cover}
{A graph $G=(V,E)$ and an integer $k$}
{Decide if there is a vertex set $S\subseteq V$ of size $k$ such that all edges are incident to at least one vertex from $S$.}

\defproblem{Connected Vertex Cover}
{A graph $G=(V,E)$ and an integer $k$}
{Decide if there is a vertex set $S\subseteq V$ of size $k$ such that $S$ induces a connected subgraph, and all edges are incident to at least one vertex from $S$.}

\defproblem{Dominating Set}
{A graph $G=(V,E)$ and an integer $k$}
{Decide if there is a vertex set $D\subseteq V$ of size $k$ such that all vertices in $V\setminus D$ are adjacent to at least one vertex in $D$.}

\defproblem{$r$-Dominating Set}
{A graph $G=(V,E)$ and an integer $k$}
{Decide if there is a vertex set $D\subseteq V$ of size $k$ such that all vertices in $V\setminus D$ have at least one vertex of $D$ within distance $r$.}

\defproblem{Connected Dominating Set}
{A graph $G=(V,E)$ and an integer $k$}
{Decide if there is a vertex set $D\subseteq V$ of size $k$ such that $D$ induces a connected subgraph, and all vertices in $V\setminus D$ are adjacent to at least one vertex in $D$.}

\defproblem{Steiner Tree}
{A graph $G=(V,E)$, a set of terminal vertices $K\subseteq V$ and integer $s$.}
{Decide if there is a vertex set $X\subseteq V$ of size at most $s$, such that $K\subseteq X$, and $X$ induces a connected subgraph of $G$.}

\defproblem{Maximum Induced Forest}
{A graph $G=(V,E)$ and an integer $k$}
{Decide if there is a vertex set $F\subseteq V$ of size $k$ such that $F$ induces a forest.}

\defproblem{Feedback Vertex Set}
{A graph $G=(V,E)$ and an integer $k$}
{Decide if there is a vertex set $F\subseteq V$ of size $k$ such that $V\setminus F$ induces a forest.}

\defproblem{Connected Feedback Vertex Set}
{A graph $G=(V,E)$ and an integer $k$}
{Decide if there is a vertex set $F\subseteq V$ of size $k$ such that $F$ induces a connected subgraph, and $V\setminus F$ induces a forest.}

\defproblem{Hamiltonian Cycle}
{A graph $G=(V,E)$}
{Decide if there is a cycle $S\subseteq E$ that visits all vertices of $G$.}

\defproblem{Hamiltonian Path}
{A graph $G=(V,E)$, and two vertices $v,w\in V$}
{Decide if there is a path $P\subseteq E$ from $v$ to $w$ in $G$ that visits all vertices of $G$.}

% %New numbering for appendix-only theorems
% \setcounter{defin}{0}
% \renewcommand{\thedefin}{\Alph{section}\arabic{defin}}
% \renewcommand{\thetheo}{\Alph{section}\arabic{defin}}
% \renewcommand{\thelem}{\Alph{section}\arabic{defin}}
% \renewcommand{\thepropo}{\Alph{proposition}\arabic{defin}}
% \renewcommand{\thecoro}{\Alph{section}\arabic{defin}}
% \renewcommand{\theobse}{\Alph{section}\arabic{defin}}
% \renewcommand{\therem}{\Alph{section}\arabic{defin}}

% \subfile{appendix_alg.tex}

% \subfile{appendix_lower.tex}

\end{document}

%% file: intro.tex
%--------------------------------------------------------------------------------------------
\section{Introduction}
%--------------------------------------------------------------------------------------------
Many hard graph problems that seem to require $2^{\Omega(n)}$ time on general graphs,
where $n$ is the number of vertices, can be solved in subexponential time on planar graphs.
In particular, many of these problems can be solved in $2^{O(\sqrt{n})}$ time on planar graphs.
Examples of problems for which this so-called \emph{square-root phenomenon}~\cite{Marx13} holds
include \textsc{Independent Set}, \textsc{Vertex Cover}, \textsc{Hamiltonian Cycle}.
The great speed-ups that the square-root
phenomenon offers lead to the question: are there other graph classes that
also exhibit this phenomenon, and is there an overarching framework
to obtain algorithms with subexponential running time for these graph classes?
The planar separator theorem~\cite{LiptonT79,LiptonT80} and treewidth-based algorithms~\cite{fptbook} offer
a partial answer to this question. They give a general framework to obtain
subexponential algorithms on planar graphs or, more generally, on $H$-minor free graphs.
It builds heavily on the fact that $H$-minor free graphs have treewidth~$O(\sqrt{n})$
and, hence, admit a separator of size $(\sqrt{n})$. A similar line of work is emerging in the area of geometric intersection graphs, with running times of the form $n^{O(n^{1-1/d})}$, or in one case  $2^{O(n^{1-1/d})}$ in the $d$-dimensional case~\cite{MarxS14,SmithW98}.
The main goal of our paper is to establish a framework for a wide class of
geometric intersection graphs that is similar to the framework known for planar graphs, while guaranteeing the running time $2^{O(n^{1-1/d})}$.

The \emph{intersection graph} $\iG{F}$ of a set $F$ of objects in $\Reals^d$ is the graph
whose vertex set is $F$ and in which two vertices are connected when the
corresponding objects intersect. \emph{(Unit-)disk graphs}, where~$F$
consists of (unit) disks in the plane are a widely studied
class of intersection graphs. Disk graphs form a natural generalization of planar
graphs, since any planar graph can be realized as the intersection graph of a set
of disks in the plane. In this paper we consider intersection graphs
of a set $F$ of \emph{fat objects}, where an
object $o\subseteq \Reals^d$ is $\alpha$-\emph{fat}, for some $0<\alpha\leq 1$
if there are balls $\Bin$ and $\Bout$ in $\Reals^d$ such that $\Bin\subseteq o \subseteq \Bout$
and $\radius(\Bin)/\radius(\Bout)\geq \alpha$.
For example, disks are 1-fat and squares are $(1/\sqrt{2})$-fat. From now on we
assume that $\alpha$ is an absolute constant, and often simply speak of fat objects.
\oldtext{Note that we do not require the objects in $F$ to be convex, or even connected.
Thus our definition is very general. In particular, it does not imply that $F$
has near-linear union complexity, as is the case for so-called
locally-fat objects~\cite{abes-union-complexity-fat}. In most of our results
we furthermore assume that the objects in $F$ are \emph{similarly sized},
meaning that the ratio of their diameters is bounded by a fixed constant.}
\newtext{When dealing with arbitrarily-sized fat objects we also require
the objects to be convex. Most of our results are about similarly-sized fat objects,
however, and then we do not need the objects to be convex; in fact, they do not
even need to be connected.\footnote{In the conference version of our
paper we erroneously claimed that for arbitrarily-sized objects the
restriction to convex objects is not necessary either.}
(A set of objects is \emph{similarly-sized} when the
ratio of the largest and smallest diameter of the objects in the set is bounded by a fixed constant.) 
Thus our definition of fatness for similarly-sized fat objects is very general.
In particular, it does not imply that $F$
has near-linear union complexity, as is the case for so-called
locally-fat objects~\cite{abes-union-complexity-fat}.
}

Several important graph problems have been investigated for (unit-)disk graphs or other types
of intersection graphs~\cite{AlberF04,BiroBMMR17,FominLPSZ17,FominLS12,MarxS14}. However,
an overarching framework that helps designing subexponential algorithms
has remained elusive. A major hurdle to obtain such a framework
is that even unit-square graphs can already have arbitrarily large cliques and so
they do not necessarily have small separators or small treewidth. One may hope
that intersection graphs have low cliquewidth or rankwidth---this has proven
to be useful for various dense graph classes \cite{CourcelleMR00,Oum17}---but unfortunately
this is not the case even when considering only unit interval graphs~\cite{GolumbicR00}.
One way to circumvent this hurdle is to restrict the attention to intersection graphs
of disks of \emph{bounded ply}~\cite{BasteT17,hpq-aapeldg-15}. 
This prevents large cliques, but the restriction to bounded-ply graphs severely limits the inputs that can be handled. A major goal of our work is thus to give a framework that can even be applied when the ply is unbounded.

%------------------------------------------------------------------------------------------
\paragraph*{Our first contribution: an algorithmic framework for geometric intersection graphs of fat objects}
%------------------------------------------------------------------------------------------
As mentioned, many subexponential results for planar graphs rely on planar separators.
Our first contribution is a generalization of this result to intersection graphs of 
\oldtext{(arbitrarily-sized)} 
\newtext{(arbitrarily-sized and convex, or similarly-sized)}
fat objects in~$\Reals^d$. Since these graphs can have large cliques we cannot bound the
number of vertices in the separator. Instead, we build a separator consisting of cliques.
We then define a weight function~$\gamma$ on these cliques, and we define the
weight of a separator as the sum of the weights of its constituent cliques
$C_i$. This is useful since for many problems a separator can intersect the
solution vertex set in $2^{O(\sum_i \gamma(|C_i|))}$ many ways, for a suitable
function~$\gamma$. Although we state
Theorem~\ref{thm:arbitrary-size-separator} with the strongest bound on
$\gamma$ possible,  in our applications it suffices to define the weight of a
clique~$C$ as $\gamma(|C|):=\log(|C|+1)$. Our theorem can now be stated as
follows.

%---------------------------------------------------------------------------------------
\begin{theorem}\label{thm:arbitrary-size-separator}
Let $d\geq 2, \alpha>0$ and $\eps>0$ be constants and let $\gamma$ be a weight function such that $\gamma(t) = O(t^{1-1/d-\eps})$. Let $F$ be a set of $n$ $\alpha$-fat objects in~$\Reals^d$
\newtext{that are all convex, or similarly-sized}. Then the intersection
graph~$\iG{F}$ has a $(6^d/(6^d+1))$-balanced separator $F_\sep$ and a clique partition
$\cC(F_\sep)$ of $F_\sep$ with weight $O(n^{1-1/d})$. Such a separator and a
clique partition $\cC(F_\sep)$ can be computed in $O(n^{d+2})$ time if the
objects have constant complexity.
\end{theorem}
%--------------------------------------------------------------------------------------

\begin{remark} The time stated in the theorem to compute the separator is $O(n^{d+2})$. It is probably possible to reduce this, but since in our applications this does not make a difference for the final time bounds, we do not pursue this.
\end{remark}

\begin{table*}[t]
\begin{center}
 \resizebox{0.95\textwidth}{!}{
  \begin{tabular}{ l  c  c  c c}
    \hline
    Problem & Algorithm class & Representation-agnostic & Lower bound class \\
    \hline
    \textsc{Independent Set} & \newtext{Convex fat} & no & Unit Ball, $d\ge 2$\\
    \textsc{Independent Set} & Sim. sized fat &yes & Unit Ball, $d\ge 2$\\
    \textsc{$r$-Dominating Set}, $r=$ const & Sim. sized fat & yes & Induced Grid, $d \ge 2$ \\
    \textsc{Steiner Tree} & Sim. sized fat &yes &  Induced Grid, $d \ge 2$\\
    \textsc{Feedback Vertex Set} & Sim. sized fat &yes & Induced Grid, $d\ge 2$ \\
    \textsc{Conn. Vertex Cover} & Sim. sized fat &yes & Unit Ball, $d\ge 2$ \textbf{or} Induced Grid, $d\ge 3$ \\
    \textsc{Conn. Dominating Set} & Sim. sized fat &yes & Induced Grid, $d\ge 2$ \\
    \textsc{Conn. Feedback Vertex Set} & Sim. sized fat &yes & Unit Ball, $d\ge 2$ \textbf{or} Induced Grid, $d\ge 3$  \\
    \textsc{Hamiltonian Cycle/Path} & Sim. sized fat & no & Induced Grid, $d \ge 2$ \\
    \hline
  \end{tabular}
 }
\end{center}
\vspace{0.5em}
\caption{Summary of our results. In each case we list the most inclusive class where our framework leads to algorithms with  $2^{O(n^{1-1/d})}$ running time, and the most restrictive class for which we have a matching lower bound. We also list whether the algorithm is representation-agnostic.}
\label{tab:results}
\end{table*}

A direct application of our separator theorem
is a $2^{O(n^{1-1/d})}$ algorithm for \textsc{Independent Set} for any fixed constant $d$. (The dependence on $d$ is double-exponential.)
For general fat objects, only the $2$-dimensional case was known to have such an algorithm~\cite{MarxP15}.

Our separator theorem can be seen as a generalization of the work of Fu~\cite{Fu2011}
who considers a weighting scheme similar to ours.
However, Fu's result is significantly less general as it only applies to unit balls
and his proof is arguably more complicated. Our result can also be seen as
a generalization of the separator theorem of Har-Peled and Quanrud~\cite{hpq-aapeldg-15}
which gives a small separator for constant ply---indeed, our proof borrows some ideas from theirs.

Finally, the technique employed by Fomin et al.~\cite{FominLPSZ17} in two dimensions has also similar qualities; in particular, the idea of using cliques as a basis for a separator can also be found there, and leads to subexponential parameterized algorithms, even for some problems that we do not tackle here.
\medskip

\oldtext{After proving the weighted separator theorem for arbitrarily-sized fat objects,
we switch to similarly-sized objects.}
\newtext{After proving the weighted separator theorem for fat objects,
we apply it to obtain an algorithmic framework for similarly-sized fat objects.}
Here the idea is as follows: We find a suitable
clique-decomposition $\cP$ of the intersection graph $G[F]$, contract each clique to
a single vertex, and then work with the contracted graph~$G_{\cP}$ where the node
corresponding to a clique~$C$ gets weight~$\gamma(|C|)$. We then prove that the
graph $G_\cP$ has constant degree and, using our separator theorem, we prove that
$G_{\cP}$ has weighted treewidth $O(n^{1-1/d})$. (The big-O notation hides an exponential factor in $d$.) Moreover, we can compute a
tree decomposition of this weight in $2^{O(n^{1-1/d})}$ time.

\begin{figure}[t]
\begin{center}
\includegraphics[width=0.6\textwidth]{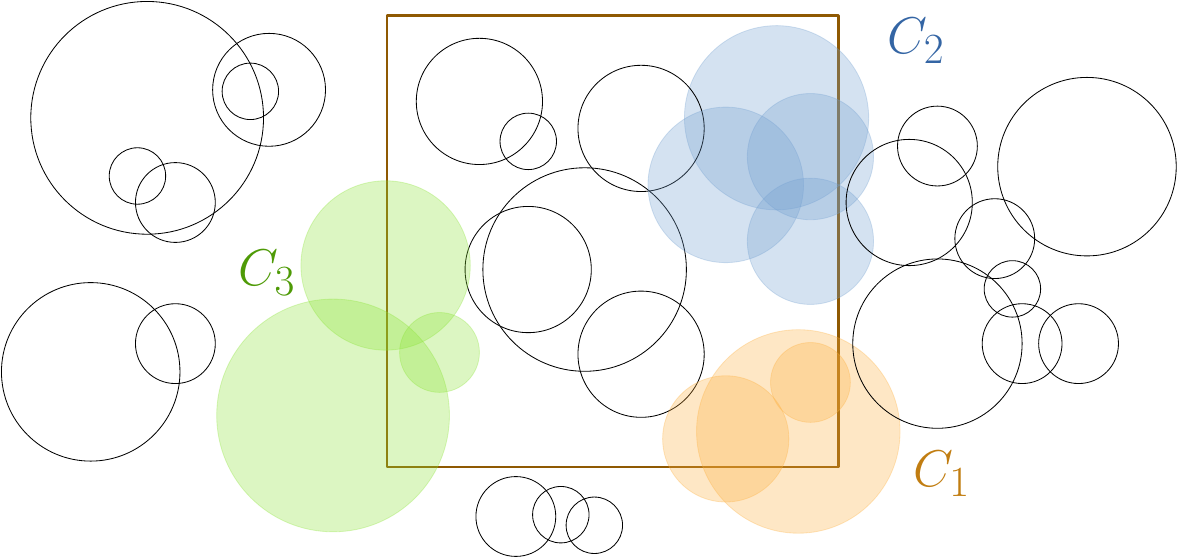}
\end{center}
\caption{An example for Theorem~\ref{thm:arbitrary-size-separator}: a disk graph with a separator partitioned into cliques $C_1,C_2,C_3$.}%
\label{fig:separator_example}
\end{figure}

Thus we obtain a framework that gives $2^{O(n^{1-1/d})}$-time algorithms for
intersection graphs of similarly-sized\footnote{With separator-based results it is often possible to state theorems for all subgraphs of a given graph class. This is not possible in our case: taking subgraphs destroys the cliques that we rely on. Additionally, every graph class we consider contains all complete graphs, so a statement about their subgraphs would have to hold for all graphs.} fat objects for many problems for which treewidth-based
algorithms are known. Our framework recovers and often slightly improves
the best known results for several problems,\footnote{Note that most of the earlier results are in the parameterized setting, but we do not consider parameterized algorithms here.} including
\textsc{Independent Set}, \textsc{Hamiltonian Cycle} and \textsc{Feedback Vertex Set}.
Our framework also gives the first subexponential algorithms in geometric intersection graphs for,
among other problems, \textsc{$r$-Dominating Set} for constant $r$, \textsc{Steiner Tree} and \textsc{Connected Dominating Set}.

Furthermore, we show that our approach can be combined with the \emph{rank-based approach}~\cite{single-exponential},
a technique to speed up algorithms for connectivity problems. Table \ref{tab:results} summarizes the results we obtain by applying our framework; in each case we have matching upper and lower bounds on the time complexity of $2^{\Theta(n^{1-1/d})}$ (where the lower bounds are conditional on the Exponential Time Hypothesis).
\medskip

A desirable property of algorithms for geometric graphs is that they are
\emph{representation-agnostic},\index{representation agnostic} meaning that they can work directly on the
graph without a geometric representation of $F$. Most of the known algorithms
do in fact require a representation, which could be a problem in applications,
since finding a geometric representation (e.g., with unit disks) of a given geometric intersection
graph is \NP-hard~\cite{BreuK98}, and many recognition problems for geometric
graphs are $\exists\Reals$-complete~\cite{Kang12}. Note that in the absence
of a representation, some of our algorithms require that the input graphs are
from the proper graph class, otherwise they might give incorrect answers.

One of the advantages of our framework is that it
yields representation agnostic algorithms for many problems. To this end we
need to generalize our scheme slightly: We no longer work with a clique
partition to define the contracted graph~$G_\cP$, but with a partition whose
classes are the union of constantly many cliques. We show that such a
partition can be found efficiently without knowing the set $F$ defining the
given intersection graph. Thus we obtain representation-agnostic algorithms
for many of the problems mentioned above, in contrast to known results which
almost all need the underlying set~$F$ as input.

%--------------------------------------------------------------------------------------------
\paragraph*{Our second contribution: a framework for lower bounds under ETH}
%--------------------------------------------------------------------------------------------
The $2^{O(n^{1-1/d})}$-time algorithms that we obtain for many problems
immediately lead to the question: is it possible to obtain even faster
algorithms? For many problems on planar graphs, and for certain problems on
ball graphs the answer is no, assuming the Exponential Time Hypothesis (ETH)~\cite{ImpagliazzoP01}.
However, these lower bound results in higher dimensions are scarce, and often
very problem-specific. Our second contribution is a framework to obtain tight
ETH-based lower bounds for problems on $d$-dimensional grid graphs (which are
a subset of intersection graphs of similarly-sized fat objects). The obtained lower
bounds match the upper bounds of the algorithmic framework. Our
lower bound technique is based on a constructive embedding of graphs into
$d$-dimensional grids, for $d\geq 3$, thus avoiding the invocation of deep
results from Robertson and Seymour's graph minor theory.  This
\emph{Cube Wiring Theorem} implies that for any constant $d\ge 3$, any connected graph on $m$
edges is the minor of the $d$-dimensional grid hypercube of side length
$O(m^{\frac{1}{d-1}})$ (see Theorem~\ref{thm:cubewiringminor}).

As it turns out, we can easily derive the Cube Wiring Theorem from a result of Thompson and Kung~\cite{DBLP:conf/stoc/ThompsonK76}. We also prove a slightly stronger version of Cube Wiring, which may be of independnet interest.

For $d=2$, we give a lower bound for a customized version of the
\textsc{3-SAT} problem. Now, these results make it possible to design simple reductions for our problems using just three custom gadgets per problem; the gadgets
model variables, clauses, and connections between variables and clauses, respectively. By invoking Cube Wiring or our custom satisfiability problem, the wires connecting the clause and variable gadgets can be routed in a very tight space.
Giving these three gadgets immediately yields the tight lower bound in $d$-dimensional grid graphs (under ETH) for all $d\ge 2$. Naturally, the same conditional lower bounds are implied in all containing graph classes, such as unit ball graphs, unit cube graphs and also in intersection graphs of similarly sized fat objects. Similar lower bounds are known for various problems in the parameterized complexity literature~\cite{MarxS14,BiroBMMR17}. The embedding in~\cite{MarxS14} in particular has a denser target graph than a grid hypercube, where the ``edge length'' of the cube contains an extra logarithmic factor compared to ours (see Theorem 2.17 in~\cite{MarxS14}) and thereby gives slightly weaker lower bounds.

Moreover, our lower bound for \HC in induced grid graphs implies the same lower bound for \etsp, which turns out to be ETH-tight~\cite{BergBKK18}.

%% file: algorithmic_framework.tex
%-----------------------------------------------
\section{The algorithmic framework}
%-----------------------------------------------
%
%\skb{Added some intro text here, it is optional for the conference version.}
%
%Our algorithmic framework is built on two general separator theorems. The
%first is a separator for general fat objects, where we define a clique cover
%of the separating vertex set, and introduce a cost on these cliques. The
%second theorem, which is a relatively easy consequence of the first only
%applies to similarly-sized fat objects, but due to various nice properties it
%is suited for a wider range of applications.
%
%After the separator theorems, we introduce weighted tree decompositions on a contracted graph, and show how such decompositions can be obtained.

%----------------------------------------------------------------------------------
\subsection{\newtext{Separators for fat objects}}
\label{se:arbitrary-size}
%----------------------------------------------------------------------------------
Let $F$ be a set of $n$ $\alpha$-fat objects in $\Reals^d$ for some constant $\alpha>0$, and let $\iG{F}=(F,E)$ be the intersection
graph induced by~$F$. We say that a subset $F_\sep\subseteq F$ is a
\emph{$\beta$-balanced separator} for $\iG{F}$ if $F\setminus F_\sep$ can be partitioned into
two subsets $F_1$ and $F_2$ with no edges between them and with $\max(|F_1|,|F_2|)\leq \beta n$.
% We will actually need a generalized version, where the balance is realized with respect to
% a given subset $W\subseteq F$, that is, we require that $\max(|F_1\cap W|,|F_2\cap W|)\leq \beta |W|$.
For a given decomposition $\cC(F_\sep)$ of $F_\sep$ into cliques
and a given weight function $\gamma$ we define
the \emph{weight} of~$F_\sep$, denoted by $\weight(F_\sep)$, as
$\weight(F_\sep) := \sum_{C\in \cC(F_\sep)} \gamma(|C|)$.
Next we prove that $\iG{F}$ admits a balanced separator of weight~$O(n^{1-1/d})$
for any cost function $\gamma(t) = O(t^{1-1/d-\eps})$ with $\eps>0$.
Our approach borrows ideas from Har-Peled and Quanrud~\cite{hpq-aapeldg-15},
who show the existence of small separators for low-density sets of objects,
although our arguments are significantly more involved.

\paragraph*{Step~1: Finding candidate separators.}
Let $H_0$ be a minimum-size hypercube containing at least
$n/(6^d+1)$ objects from $F$,
and assume without loss of generality that $H_0$ is the unit hypercube centered at the origin.
Let $H_1,\ldots,H_{m}$ be a collection of $m:=n^{1/d}$ hypercubes, all centered
at the origin, where $H_i$ has edge length $1+\frac{2i}{m}$. Note that the largest
hypercube, $H_m$, has edge length~3, and that the distance between the corresponding faces of consecutive hypercubes
$H_i$ and $H_{i+1}$ is $1/n^{1/d}$.

Each hypercube~$H_i$ induces a partition of $F$ into three subsets:
a subset~$F_{\myin}(H_i)$ containing all objects \newtext{whose convex hull} lies 
completely in the interior of~$H_i$,
a subset~$F_{\bd}(H_i)$ containing all objects \newtext{whose convex hull} intersects the boundary $\bd H_i$ of~$H_i$,
and a subset~$F_{\myout}(H_i)$ containing all objects \newtext{whose convex hull} 
lies completely in the exterior of~$H_i$.
Obviously an object from $F_{\myin}(H_i)$ cannot intersect an object from $F_{\myout}(H_i)$,
and so $F_{\bd}(H_i)$ defines a separator in a natural way.
\newtext{(Note that disconnected objects could lie partly
inside $H_i$ and partly outside $H_i$ without intersecting $\bd H_i$. This is the reason why we define the sets 
$F_{\myout}(H_i)$, $F_{\bd}(H_i)$, and  $F_{\myin}(H_i)$ with respect to the 
convex hulls of the objects. If each object is connected, we can also define
these sets with respect to the objects themselves.)}
It will be convenient to add some more objects to these separators, as follows.
We call an object \emph{large} when its diameter is at least~$1/4$,
and \emph{small} otherwise. We will add all large objects that intersect~$H_m$ to our
separators. Thus our candidate separators are the sets
$F_{\sep}(H_i) := F_{\bd}(H_i) \cup F_\mylarge$, where
$F_{\mylarge}$ is the set of all large objects intersecting~$H_m$.
We show that our candidate separators are balanced:
%---------------------------------------------------------------------------------------
\begin{lemma}\label{le:arbitrary-size-balance}
For any $0\leq i\leq m$ we have
\[
\max\big(|F_{\myin}(H_i)\setminus F_\mylarge|,|F_{\myout}(H_i)\setminus F_\mylarge| \big) < \frac{6^d}{6^d+1}n.
\]
\end{lemma}
\vspace{-1em}

\begin{proof}
Consider a hypercube~$H_i$. Because $H_0$ contains at least $n/(6^d+1)$ objects from $F$, we immediately obtain
\[
\big|(F_{\myout}(H_i)\setminus F_\mylarge)\big|
    \leq    |F_{\myout}(H_0)|
    \leq    |F \setminus F_{\myin}(H_0)|
    <    \left(1-\frac{1}{6^d+1}\right)n
    =   \frac{6^d}{6^d+1}n.
\]
To bound $\big|F_{\myin}(H_i)\setminus F_\mylarge\big|$, consider a subdivision of
$H_i$ into $6^d$ sub-hypercubes of edge length
$\frac{1}{6}(1+\frac{2i}{m}) \leq 1/2$. We claim that any sub-hypercube $H_{\sub}$ intersects
fewer than $n/(6^d+1)$ small objects from $F$. To see this, recall that small objects have
diameter less than~1/4. Hence, all small objects intersecting $H_{\sub}$ are fully
contained in a hypercube of edge length less than~1. Since $H_0$ is a smallest
hypercube containing at least $n/(6^d+1)$ objects from $F$,
$H_{\sub}$ must thus intersect fewer than $n/(6^d+1)$ objects from $F$, as claimed.
Each object in $F_{\myin}(H_i)$ intersects at least one of the
$6^d$ sub-hypercubes, so we can conclude that
$\big |F_{\myin}(H_i)\setminus F_\mylarge \big| < \big(6^d/(6^d+1)\big)n$.
\end{proof}
%---------------------------------------------------------------------------------------

\paragraph*{Step~2: Defining the cliques and finding a low-weight separator.} Define $F^* := F \setminus (F_{\myin}(H_0) \cup F_{\myout}(H_m)\cup F_{\mylarge})$. Note that
$F_{\bd}(H_i)\setminus F_{\mylarge}\subseteq F^*$ for all~$i$. We partition $F^*$ into \emph{size classes}~$F^*_s$, based
on the diameter of the objects. More precisely, for integers $s$ with $1\leq s \leq s_{\max}$,
where $s_{\max} := \ceil{(1-1/d)\log n}-2$, we define
\[
F^*_s := \left\{ o \in F^*: \frac{2^{s-1}}{n^{1/d}} \leq \diam(o) < \frac{2^s}{n^{1/d}} \right\}.
\]
We furthermore define $F^*_0$ to be the subset of objects $o\in F^*$ with $\diam(o)<1/n^{1/d}$.
Note that $2^{s_{\max}}/n^{1/d} \geq 1/4$, which means that
every object in $F^*$ is in exactly one size class.

Each size class other than $F^*_0$ can be decomposed into cliques as follows: Fix a size class~$F^*_s$, with $1\leq s\leq s_{\max}$.
Since the objects in $F$ are $\alpha$-fat for a fixed constant~$\alpha>0$, each $o\in F^*_s$ contains a ball
of radius $\alpha\cdot(\diam(o)/2) = \Omega(\frac{2^{s}}{n^{1/d}})$.
Moreover, each object $o\in F^*_s$ lies fully or partially inside the outer hypercube $H_m$,
which has edge length~3. This implies we can stab all objects in
$F^*_s$ using a set $P_s$ of $O((\frac{n^{1/d}}{2^{s}})^d)$ points.
%  (For example, we enlarge $H_m$ by 1/4 on all sides so that the enlarged hypercube
%  contains all objects in $F^*_s$, put a grid of whose cells have edge
%  length~$c_d (\frac{n^{1/d}}{2^{s}})^d$ for an appropriate constant $c_d$ depending
%  on the dimension, and use the grid points as stabbing points.)
Thus there exists a decomposition
$\cC(F^*_s)$ of $F^*_s$ consisting of $O(\frac{n}{2^{sd}})$ cliques.
\oldtext{
In a similar way we can argue that there exists
a decomposition $\cC(F_{\mylarge})$ of $F_{\mylarge}$ into $O(1)$ cliques.}
\newtext{Next we show that there exists
a such a decompostion for $F_{\mylarge}$ as well.}

\begin{lemma}
Let $F$ be a set of similarly-sized fat objects or a set of arbitrarily-sized 
convex fat objects. Then $F_{\mylarge}$ can be decomposed into a collection
$\cC(\Flarge)$ of $O(1)$ cliques.
\end{lemma}
\begin{proof}
Recall that $H_0$ is a smallest hypercube containing at least $n/(6^d+1)$ objects, 
and that we assumed $H_0$ to be a unit hypercube centered at the origin.
Let $H(t)$ denote a copy of $H_0$ scaled by a factor~$t$ with respect to the origin.
Note that $H_0= H(1)$ and $H_m = H(3)$. 

To prove the lemma for the case of similarly-sized objects, let $d_{\min}$
and $d_{\max}$ be the minimum and maximum diameter of any of the
objects in $F$, respectively. Since $H_0$ has unit size and fully contains at least one object from $F$,
we have $d_{\min} \leq \sqrt{d} = O(1)$. Moreover, the objects
are similarly sized and so we also have $d_{\max} = O(1)$.
Because all objects in $\Flarge$ intersect $H_m=H(3)$, they must lie completely inside 
$H(t^*)$ for $t^*=3+2d_{\max}=O(1)$. Since $\diam(o)\geq 1/4$ for all $o\in \Flarge$
and the objects are fat, each object contains a ball of radius $\Omega(1)$.
Hence, we can stab all objects in $\Flarge$ with a grid
of $O(1)$ points inside $H(t^*)$.

\begin{figure}[t]
\begin{center}
\includegraphics{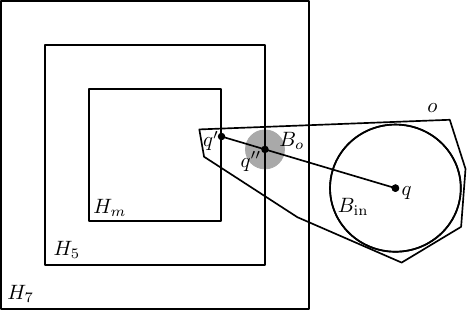}
\end{center}
\caption{A convex fat object $o\in \Flarge$ contains a ball $B_o$ of at least constant radius that is inside the hypercube $H_7$.}%
\label{fig:Flarge_lemma}
\end{figure}

Next we prove the lemma for the case where the objects in $\Flarge$ 
are arbitrarily-sized but convex. It suffices to show that for any
$o\in \Flarge$ there exists a ball $B_o\subseteq o$ of radius $\Theta(1)$
that is fully contained in $H(7)$. To show the existence
of $B_o$, let $\Bin \subseteq o$ be a ball of radius $\Theta(\diam(o))$; 
such a ball exists because $o$ is fat.
Note that $\diam(o)\geq 1/4$ since $o\in \Flarge$.
Let $q$ be the center of $\Bin$, see Fig.~\ref{fig:Flarge_lemma}.
If $q\in H(5)$ then the ball centered at $q$ of radius
$\min(\radius(\Bin),1)$ is a ball of radius $\Theta(1)$ that lies
completely inside $H(7)$ and we are done. Otherwise, let
$q'\in o\cap H_m$ and let $q'' := q q' \cap \bd H(5)$ be the point
where the segment $qq'$ intersects $\bd H(5)$.
We now take $B_o$ to be the ball centered at $q''$ and
of radius $\min\left(1,\frac{|q'q''|}{|q'q|}\cdot\radius(\Bin)\right)$.
By convexity of $o$, we have $B_o\subset o$.
Moreover, $B_o\subset H(7)$.
Finally, since $|q'q''|\geq 1$ and $q'q \leq \diam(o)$ we have
\[
\frac{|q'q''|}{|q'q|} \cdot\radius(\Bin)
    \geq \frac{1}{\diam(o)} \cdot\radius(\Bin)
    = \Theta(1)
\]
which shows $\radius(B_o) =\Theta(1)$. 
This finishes the proof that $B_o$ has the desired properties.
\end{proof}
\oldtext{For $F^*_0$ the argument does not work}
\newtext{The set $F^*_0$ cannot be decomposed into few cliques}
since objects in $F^*_0$ can be arbitrarily
small. Hence, we create a singleton clique for each object in $F^*_0$.
Together with the decompositions of the size classes~$F^*_s$ and of $\Flarge$ we thus
obtain a decomposition $\cC(F^*)$ of $F^*$ into cliques.

Note that $\cC(F^*)$ induces a decomposition of $F_{\sep}(H_i)$  into cliques,
for any $i$. We denote this decomposition by $\cC(F_{\sep}(H_i))$. Thus, for a
given weight function~$\gamma$, the weight of $F_{\sep}(H_i)$ is $\sum_{C\in
\cC(F_{\sep}(H_i))} \gamma(|C|)$. Our goal is now to show that at least one of
the separators $F_{\sep}(H_i)$ has weight $O(n^{1-1/d})$, when $\gamma(t) =
O(t^{1-1/d-\eps})$ for some $\eps>0$. To this end we will bound the total
weight of all separators $F_{\sep}(H_i)$ by $O(n)$. Using that the number of
separators is $n^{1/d}$ we then obtain the desired result.
%---------------------------------------------------------------------------------------
\begin{lemma}\label{le:arbitrary-size-weight}
If $\gamma(t) = O(t^{1-1/d-\eps})$ for some $\eps>0$ then $\sum_{i=1}^m \weight(F_{\sep}(H_i)) = O(n)$.
\end{lemma}
%---------------------------------------------------------------------------------------
\vspace{-1em}
\begin{proof}
First consider the cliques in $\cC(F^*_0)$, which are singletons. Since objects in $F^*_0$ have
diameter less than $1/n^{1/d}$, which is the distance between consecutive hypercubes~$H_i$ and $H_{i+1}$,
each such object is in at most one set $F_{\bd}(H_i)$. Hence, its contribution
to the total weight $\sum_{i=1}^m \weight(F_{\sep}(H_i))$ is $\gamma(1)=O(1)$. Together, the cliques in $\cC(F^*_0)$
thus contribute $O(n)$ to the total weight.

Next, consider $\cC(\Flarge)$. It consists of $O(1)$ cliques. In the worst case each clique
appears in all sets $F_{\bd}(H_i)$. Hence, their total contribution to
$\sum_{i=1}^m \weight(F_{\sep}(H_i))$ is bounded by $O(1) \cdot \gamma(n) \cdot n^{1/d} = O(n)$.

Now consider a set $\cC(F^*_s)$ with $1\leq s\leq s_{\max}$.
A clique $C\in \cC(F^*_s)$ consists of objects of diameter at most $2^s/n^{1/d}$
that are stabbed by a common point. Since the distance between consecutive hypercubes
$H_i$ and $H_{i+1}$ is $1/n^{1/d}$, this implies that $C$
contributes to the weight of $O(2^s)$ separators~$F_{\sep}(H_i)$. The contribution
to the weight of a single separator is at most $\gamma(|C|)$. (It can be less than $\gamma(|C|)$ because not all
objects in $C$ need to intersect $\bd H_i$.) Hence, the total weight contributed
by all cliques, which equals the total weight of all separators, is
% \begin{align*}
% \sum_{s=1}^{s_{\max}} \sum_{C\in \cC(F^*_s)} &(\mbox{weight contributed by~$C$})\\
%     &\leq \sum_{s=1}^{s_{\max}} \sum_{C\in \cC(F^*_s)} 2^s \gamma(|C|)\\
%     &= \sum_{s=1}^{s_{\max}} \left( 2^s \sum_{C\in \cC(F^*_s)} \gamma(|C|) \right).
% \end{align*}

\[
\sum_{s=1}^{s_{\max}} \sum_{C\in \cC(F^*_s)} (\mbox{weight contributed by~$C$})
    \leq \sum_{s=1}^{s_{\max}} \sum_{C\in \cC(F^*_s)} 2^s \gamma(|C|)
    = \sum_{s=1}^{s_{\max}} \left( 2^s \sum_{C\in \cC(F^*_s)} \gamma(|C|) \right).
\]

Next we wish to bound $\sum_{C\in \cC(F^*_s)} \gamma(|C|)$. Define $n_s := |F^*_s|$
and observe that $\sum_{s=1}^{s_{\max}} n_s \leq n$. Recall that $\cC(F^*_s)$ consists of
$O(n/2^{sd})$ cliques, that is, of at most $cn/2^{sd}$ cliques for some constant~$c$.
To make the formulas below more readable we assume $c=1$ (so we can omit~$c$), but it
is easily checked that this does not influence the final result asymptotically.
Similarly, we will be using $\gamma(t) = t^{1-1/d-\eps}$ instead of $\gamma(t) = O(t^{1-1/d-\eps})$.
Because $\gamma$ is positive and concave, the sum $\sum_{C\in \cC(F^*_s)} \gamma(|C|)$ is
maximized when the number of cliques is maximal, namely $\min(n_s,n/2^{sd})$,
and when the objects are distributed as evenly as possible over the cliques.
Hence,
\[
\sum_{C\in \cC(F^*_s)} \hspace{-0.4em} \gamma(|C|) \hspace{-0.3em}\
    \leq \hspace{-0.3em} \ \left\{ \hspace{-0.3em} \begin{array}{ll}
                    n_s  & \mbox{if $n_s \leq n/2^{sd}$} \\
                    (n/2^{sd}) \cdot \gamma\left(  \frac{n_s}{n/2^{sd}} \right) & \mbox{otherwise}
                    \end{array}
             \right.
\]
We now split the set $\{1,\ldots,s_{\max}\}$ into two index sets $S_1$ and $S_2$,
where $S_1$ contains all indices~$s$ such that $n_s \leq n/2^{sd}$, and
$S_2$ contains all remaining indices. Thus
\begin{equation}
\sum_{s=1}^{s_{\max}} \left( 2^s \sum_{C\in \cC(F^*_s)} \gamma(|C|) \right)
      =  \sum_{s\in S_1} \left( 2^s \sum_{C\in \cC(F^*_s)} \gamma(|C|) \right)
         + \sum_{s\in S_2} \left( 2^s \sum_{C\in \cC(F^*_s)} \gamma(|C|) \right)
\label{eq:S1S2}
\end{equation}
The first term in (\ref{eq:S1S2}) can be bounded by
\begin{equation*}
\sum_{s\in S_1} \left( 2^s \sum_{C\in \cC(F^*_s)} \gamma(|C|) \right)
    \leq  \sum_{s\in S_1} 2^s n_s
    \leq  \sum_{s\in S_1}  2^s (n/2^{sd})
    =   n \sum_{s\in S_1}  1/2^{s(d-1)}
    =  O(n),
\end{equation*}
where the last step uses that $d\geq 2$.
For the second term we get
% \begin{array}{lll}
\begin{align*}
\sum_{s\in S_2} \left( 2^s \sum_{C\in \cC(F^*_s)} \gamma(|C|) \right)
    & \leq \sum_{s\in S_2} \left( 2^s (n/2^{sd}) \cdot \gamma\left(  \frac{n_s}{n/2^{sd}} \right) \right) \\
    & \leq \sum_{s\in S_2} \left( \frac{n}{2^{s(d-1)}} \cdot \left(  \frac{n_s 2^{sd}}{n} \right)^{1-1/d-\eps} \right) \\
    & \leq n \sum_{s\in S_2} \left( \frac{n_s}{n} \right)^{1-1/d-\eps} \frac{1}{{2^{sd\eps}}}  \\
    & \leq n \sum_{s\in S_2} \left( \frac{1}{{2^{d\eps}}} \right)^s  \\
    &   =  O(n).\qedhere
\end{align*}
% \end{array}

\end{proof}
%---------------------------------------------------------------------------------------
We are now ready to prove Theorem~\ref{thm:arbitrary-size-separator}.

%---------------------------------------------------------------------------------------
\begin{proof}[Proof of Theorem~\ref{thm:arbitrary-size-separator}]
Each candidate separator~$F_{\sep}(H_i)$ is $(6^d/(6^d+1))$-balanced by
Lemma~\ref{le:arbitrary-size-balance}. Their total weight is
$O(n)$ by Lemma~\ref{le:arbitrary-size-weight}, and since we have $n^{1/d}$
candidates one of them must have weight $O(n^{1-1/d})$. Finding this separator
can be done in $O(n^{d+2})$ time by brute force.
Indeed, to find the hypercube $H_0 = [x_1,x'_1]\times \cdots \times [x_d,x'_d]$
in $O(n^{d+2})$ time we first guess the object defining $x_i$, for all $1\leq
i\leq d$, then guess the object defining $x'_1$ (and, hence, the size of the
hypercube), and finally determine the number of objects inside the hypercube.
Once we have $H_0$, we can generate the hypercubes $H_1,\cdots,H_{n^{1/d}}$,
generate the cliques as described above, and then compute the weights of the
separators~$F_{\sep}(H_i)$ by brute force within the same time bound.
\end{proof}
%---------------------------------------------------------------------------------------

%-----------------------------------------------------------------------------
\begin{corollary}\label{thm:arbitrary-size-indep-set}
Let $F$ be a set of $n$ fat objects in~$\Reals^d$ \newtext{that are
either convex or similarly-sized}, where $d$ is a constant.
Then \IndS on the intersection graph $\iG{F}$ can be solved in $2^{O(n^{1-1/d})}$ time.
\end{corollary}
%-----------------------------------------------------------------------------

\begin{proof}
Let $\gamma(t) \eqdef \log (t+1)$, and compute a separator $F_\sep$ for $\iG{F}$
using Theorem~\ref{thm:arbitrary-size-separator}. For each subset $S_\sep
\subseteq F_\sep$ of independent (that is, pairwise non-adjacent) vertices we
find the largest independent set $S$ of $G$ such that $S\supseteq S_\sep$, by
removing the closed neighborhood of $S_\sep$ from~$G$ and recursing on the
remaining connected components. Finally, we report the largest of all these
independent sets. Because a clique $C\in \cC(F_\sep)$ can contribute at most
one vertex to $S_\sep$, we have that the number of candidate sets $S_\sep$ is
at most

\[\prod_{C \in \cC(F_\sep)} (|C| +1)
= 2^{\sum_{C \in \cC(F_\sep)} \log (|C|+1)} = 2^{O(n^{1-1/d})}. \]

Since all components on which we recurse have at most $(6^d/(6^d+1))n$
vertices, the running time $T(n)$ satisfies

\[T(n) = 2^{O(n^{1-1/d})}T\left(\frac{6^d}{6^d+1}n\right) + 2^{O(n^{1-1/d})},\]
which solves to $T(n)=2^{O(n^{1-1/d})}$.
\end{proof}

%----------------------------------------------------------------------------------
\subsection{An algorithmic framework for similarly-sized fat objects}
\label{se:similar-size}
%----------------------------------------------------------------------------------
We restrict our attention to \emph{similarly-sized} fat objects. More precisely,
we consider intersection graphs of sets $F$ of objects such that, for each $o\in F$,
there are balls $\Bin$ and $\Bout$ in $\Reals^d$ such that $\Bin\subseteq F \subseteq \Bout$,
and $\radius(\Bin)=\alpha$ and $\radius(\Bout) = 1$ for some fatness constant~$\alpha>0$.
The restriction to
similarly-sized objects makes it possible to construct a clique cover of $F$ with the
following property: if we consider the intersection graph $\iG{F}$ where the cliques are contracted
to single vertices, then the contracted graph has constant degree. Moreover, the contracted graph
admits a tree decomposition whose weighted treewidth is $O(n^{1-1/d})$. This tool
allows us to solve many problems on intersection graphs of similarly-sized fat objects.

Our tree-decomposition construction uses the separator theorem from the
previous subsection. That theorem also states that we can compute the
separator for~$\iG{F}$ in polynomial time, provided we are given~$F$. However,
finding the separator if we are only given the graph and not the underlying
set~$F$ is not easy. Note that deciding whether a graph is a unit-disk graph
is already $\exists\Reals$-complete~\cite{Kang12}. Nevertheless, we show that for
similarly-sized fat objects we can find certain tree decompositions with the
desired properties, purely based on the graph~$\iG{F}$.
% To emphasize that we do not need the set~$F$, we will from now on denote the given intersection
% graph simply by~$G$, although in some proofs we still use $F$ to indicate the (unknown) underlying set
% of similarly-sized fat objects.

%----------------------------------------------------------------------------------
\paragraph*{$\kappa$-partitions, $\cP$-contractions, and separators.}
%----------------------------------------------------------------------------------
Let $G=(V,E)$ be the intersection graph of an (unknown) set~$F$
of similarly-sized fat objects, as defined above.
The separators in the previous section use cliques as basic components.
We need to generalize this slightly, by allowing
connected unions of a constant number of cliques as basic components.
Thus we define a \emph{$\kappa$-partition} of $G$ as
a partition $\cP=(V_1,\dots,V_k)$ of~$V$ such that every partition class~$V_i$
induces a connected subgraph that is the union of at most $\kappa$~cliques.
Note that a 1-partition corresponds to a clique cover of~$G$. A natural way to define the weight of a partition class $V_i$ would be the sum of the weights of the cliques contained in it. However, it will be more convenient to define the weight of a partition class $V_i$ to be $\gamma(|V_i|)$. Since $\kappa$ is a constant, this is within a constant factor of the more natural weight.

 %In our applications, $$
%    \mdb{Do we want to add the following here?
%    Dividing a partition class further into a small number of cliques can be a daunting task,
%    as it is equivalent to the coloring problem of the complement graph, but fortunately
%    this is not needed in our framework.}

Given a $\kappa$-partition $\cP$ of~$G$ we define the \emph{$\cP$-contraction} of $G$,
denoted by $G_{\cP}$, to be the graph obtained by contracting all partition classes
$V_i$ to single vertices and removing loops and parallel edges.
In many applications it is essential that the $\cP$-contraction we work
with has maximum degree bounded by a constant.
From now on, when we speak of the degree of a $\kappa$-partition~$\cP$
we refer to the degree of the corresponding $\cP$-contraction.
% We say that a \emph{$\kappa$-partition $\cP$ has degree $\Delta$} if the
% maximum degree of $G_\cP$ is at most $\Delta$. Alternatively, we can say
% that $\cP$ is a $(\kappa,\Delta)$-partition of $G$.

The following theorem is very similar to
Theorem~\ref{thm:arbitrary-size-separator},
but it applies only for similarly-sized objects
because of the degree bound on $G_\cP$. The other main difference is that the
separator is defined on the $\cP$-contraction of a given $\kappa$-partition,
instead of on the intersection graph~$G$ itself. The statement is purely
existential; we prove a more constructive theorem in Section~\ref{sec:septw}.
%----------------------------------------------------------------------------------
\begin{theorem}\label{thm:gensep}
Let $d\geq 2$ and~$\eps>0$ be constants and let $\gamma$ be a weight function such that $\gamma(t) = O(t^{1-1/d-\eps})$. Let $G=(V,E)$ be the intersection graph of a set of $n$ similarly-sized fat objects in~$\Reals^d$.
Suppose we are given a $\kappa$-partition $\cP$ of $G$ such that $G_\cP$ has maximum degree at most $\Delta$, where~$\kappa$ and $\Delta$ are constants.
Then there exists a $(6^d/(6^d+1))$-balanced separator for $G_{\cP}$ of weight $O(n^{1-1/d})$.
% Such a separator can be computed in $O(n^{d+2})$ time, assuming the
% objects have constant complexity.
\end{theorem}
%----------------------------------------------------------------------------------

\begin{proof}[Proof sketch]
Within each class $C\in\cP$, we replace each object in $C$ with the larger
object $\bigcup_{o\in C} o$. Notice that these new objects are similarly sized
and fat (with worse constants). The new objects define a supergraph $G'$, which is also an intersection graph of similarly sized fat
objects; moreover, $\cP$ is a clique-partition of $G'$. By the condition that
$G_\cP$ has maximum degree at most $\Delta$, any new object intersects at most
$\Delta$ different new objects, that is, if we were to remove all duplicate
objects from the family, then the resulting family would have ply at most
$\Delta+1$.

The arguments in the proof of Theorem~\ref{thm:arbitrary-size-separator} can
be applied for $G'$: the clique partition $\cP^*$ is created by using stabbing
points. Each class of $\cP^*$ can contain objects from at most $\Delta+1$
classes of $\cP$. We can also make sure that $\cP^*$ is a partition that is a
coarsening of $\cP$. By Theorem~\ref{thm:arbitrary-size-separator}, the graph
$G'$ has a weighted separator (with respect to $\cP^*$) of weight
$O(n^{1-1/d})$. Since $\Delta=O(1)$ and each class of $\cP^*$ contains objects from at most $\Delta+1$ classes of $\cP$, this converts into a weighted separator
wrt. $\cP$ of weight $O(n^{1-1/d})$. Since $G$ is a subgraph of $G'$, this is
also a weighted separator of $G$ with respect to $\cP$, and it has the desired
weight.
\end{proof}

The following lemma shows that a partition $\cP$ as needed in Theorem~\ref{thm:gensep} can be computed even in the absence of geometric information. Such partitions can be computed in a greedy manner, as explained next.

\begin{definition}[Greedy partition]\label{def:greedypart}
Given a graph $G$, a partition $\cP$ of $V(G)$ is a greedy partition if
there is a maximal independent set $S$ of $G$ such that each partition class $C\in \cP$
contains exactly one vertex $v_C$ of $S$ together with some neighbors of~$v_C$.
\end{definition}
%By using an inclusion-wise maximal independent set as anchor, and assigning each vertex to a neighbor in the independent set, we can create the desired partition, at the expense of each partition class consisting of a (constant) number of cliques.
%----------------------------------------------------------------------------------
\begin{lemma}\label{lem:getpartition}
Let $d\geq 2$ be a constant. Then there exist constants $\kappa$ and $\Delta$ such that for any intersection graph $G=(V,E)$ of an (unknown) set of $n$ similarly-sized fat objects in~$\Reals^d$, a greedy $\kappa$-partition $\cP$ for which $G_{\cP}$ has maximum degree $\Delta$ can be computed in polynomial time.

%Let $G=(V,E)$ be the intersection graph of an (unknown) set of $n$ similarly-sized fat objects in~$\Reals^d$ for some constant~$d\geq 2$. Then there exist constants $\kappa$ and $\Delta$ such that a $\kappa$-partition $\cP$ for which $G_{\cP}$ has maximum degree $\Delta$ can be computed in polynomial time.
\end{lemma}
%----------------------------------------------------------------------------------
\begin{proof}
Let $S\subseteq V$ be a maximal independent set in~$G$ (i.e., it is inclusion-wise maximal). We assign each vertex
$v\in V\setminus S$ to an arbitrary vertex $s\in S$ that is a neighbor of~$v$;
such a vertex $s$ always exists since $S$ is maximal.
For each vertex $s\in S$ define
$V_s := \{s\} \cup \{v\in V\setminus S: \mbox{$v$ is assigned to $s$} \}$.
We prove that the partition $\cP := \{ V_s : s\in S\}$, which
can be computed in polynomial time, has the desired properties.

Let $o_v$ denote the (unknown) object corresponding to a vertex~$v\in V$,
and for a partition class $V_s$ define $\union(V_s) := \bigcup_{v\in V_s} o_v$.
We call $\union(V_s)$ a \emph{union-object}. Let $\cU_S := \{ \union(V_s) : s\in S\}$.
Because the objects defining $G$ are similarly-sized and fat,
there are balls $B_{\myin}(o_v)$ of radius~$\alpha=\Omega(1)$ and $B_{\myout}(o_v)$
of radius~1 such that $B_{\myin}(o_v) \subseteq o_v \subseteq B_{\myout}(o_v)$.
 
Now observe that each union-object $\union(V_s)$ is contained in a ball of radius~3.
% because $o_v$ intersects $o_s$ for each $v\in V_s$.
Hence, we can stab all balls $B_{\myin}(o_v)$, $v\in V_s$ using $O(1)$ points,
which implies that $\cP$ is a $\kappa$-partition for some $\kappa=O(1)$.

To prove that the maximum degree of $G_\cP$ is $O(1)$, we note that  
any two balls $B_{\myin}(s)$, $B_{\myin}(s')$ with $s,s'\in S$ are disjoint 
(because $S$ is an independent set in $G$). 
Since all union-objects $\union(s')$ that intersect $\union(s)$ are contained in a ball
of radius~9, an easy packing argument now shows that $\union(s)$
intersects~$O(1)$ union-objects~$\union(s)$. Hence, the node in
 $G_\cP$ corresponding to $V_s$ has degree~$O(1)$.
\end{proof}

Note that Lemma~\ref{lem:getpartition} requires the promise that the input graph is indeed an intersection graph, since otherwise the greedy partition may not be a $\kappa$-partition.

%------------------------------------------------------------------
\subsection{From separators to \texorpdfstring{$\cP$}{P}-flattened treewidth}\label{sec:septw}
%------------------------------------------------------------------
Recall that a \emph{tree decomposition} of a graph $G=(V,E)$
is a pair $(T,\sig)$ where $T$ is a tree and $\sig$ is a mapping from the vertices of $T$ to subsets of $V$
called \emph{bags}, with the following properties. Let $\bags(T,\sig) := \{
\sig(u): u \in V(T) \}$ be the set of bags associated to the vertices of $T$. Then
we have: (1) For any vertex $u\in V$ there is at least one bag in $\bags(T,\sig)$ containing it. (2) For any edge $(u,v)\in E$ there is at least one bag in $\bags(T,\sig)$ containing both $u$ and $v$. (3) For any vertex $u\in V$ the collection of bags in $\bags(T,\sig)$ containing $u$ forms a subtree of~$T$.

The \emph{width} of a tree decomposition is the size of its
largest bag minus~1, and the \emph{treewidth} of a graph $G$ equals the minimum width of a tree
decomposition of $G$. We will need the notion of {\em weighted treewidth}~\cite{EijkhofBK07}.
Here each vertex has a weight, and the
{\em weighted width} of a tree decomposition is the maximum over the bags of the
sum of the weights of the vertices in the bag (note: without the $-1$).
The {\em weighted treewidth} of a graph is the minimum weighted width over its tree decompositions.

Now let $\cP=(V_1,\ldots,V_k)$ be a $\kappa$-partition of a given graph $G$ which is the intersection
graph of similarly-sized fat objects, and let $\gamma$ be a given weight function on partition classes.
We apply the concept of weighted treewidth
to $G_{\cP}$, where we assign each vertex $V_i$ of $G_\cP$ a weight~$\gamma(|V_i|)$.
Because we have a separator for $G_\cP$ of low weight by Theorem~\ref{thm:gensep},
we can prove a bound on the weighted treewidth of~$G_\cP$ using
standard techniques.
%------------------------------------------------------------------
\begin{lemma}\label{lem:gettreewidthbound}
Let $\cP$ be a $\kappa$-partition of a family of similarly-sized fat objects
such that $G_\cP$ has maximum degree at most $\Delta$,
where~$\kappa$ and $\Delta$ are constants.
Then the weighted treewidth of $G_\cP$ is $O(n^{1-1/d})$ for any weight function $\gamma$ with $\gamma(t) = O(t^{1-1/d-\eps})$.
\end{lemma}

\begin{proof}
The lemma follows from Theorem~\ref{thm:gensep} by a minor variation on
standard techniques---see for example~\cite[Theorem 20]{Bodlaender98}.
Take a separator $S$ of $G_\cP$ as indicated by Theorem~\ref{thm:gensep}.
Recursively, make tree decompositions of the connected components
of $G_\cP \setminus S$. Take the disjoint union of these tree decompositions,
add edges to make the disjoint union connected and then add $S$ to all bags.
We now have a tree decomposition of~$G_\cP$.
As base case, when we have a subgraph of $G_\cP$ with
total weight $O(n^{1-1/d})$, then we take one bag with all vertices in this subgraph.

The weight of bags for subgraphs of $G_\cP$ with $r$ vertices fulfills
$w(r) = O(r^{1-1/d}) + w\left(\frac{6^d}{6^d+1} r\right)$, which gives that the weighted
width of this tree decomposition is $w(n)=O(n^{1-1/d})$.
\end{proof}

In all of our applications, we can fix $\gamma(x)=\log(x+1)$. For a partition $\cP$, we define the \emph{$\cP$-flattened treewidth} of $G$ as the weighted treewidth of $G_\cP$ under the weighting $\gamma(x)=\log(x+1)$.

%------------------------------------------------------------------

A \emph{blowup}\index{blowup} of a vertex $v$ by an integer $t$ results in a graph where we
replace the vertex $v$ with a clique of size $t$ (called the clique of $v$),
in which we connect every vertex to the neighborhood of $v$. Note that when
blowing up multiple vertices of a graph in succession, the resulting graph
does not depend on the order of blowups.

Consider the following algorithm to compute a weighted tree decomposition of
graph $G$ with weight function $w$.
\begin{enumerate}
\item Construct an unweighted graph $H$ by blowing up each vertex $v$ of $G$
by $w(v)$. Let $H(v)$ denote the vertices of $H$ that were gained from blowing
up $v$.
\item Compute a tree decomposition of $H$, denoted by $(T_H,\sigma_H)$.
\item Construct a tree decomposition $(T_G\sim T_H,\sigma_G)$ using the same
tree layout the following way: a vertex $v\in G$ is added to a bag if and only
if the corresponding bag in $T_H$ contains all vertices of $H(v)$.
\end{enumerate}

\begin{lemma}\label{lem:weighttoclique}
The weighted width of $(T_G,\sigma_G)$ is at most the width of
$(T_H,\sigma_H)$ plus $1$; furthermore, the weighted treewidth of $G$ is equal
to $1$ plus the treewidth of $H$.
\end{lemma}

\begin{proof}
The proof is a simple modification of folklore insights on treewidth; for
related results see~\cite{BodlaenderR03,BodlaenderF05}. The proof relies on
the following well-known observation~\cite{BodlaenderM93}.

\begin{quotation}
\noindent {\sc Observation.} \textit
Let $W\subseteq V$ form a clique in $G=(V,E)$. Then all tree decompositions
$(T,\sigma)$ of $G$ have a bag $\sigma(u)\in \bags(T,\sigma)$  with
$W\subseteq \sigma (u)$.
\end{quotation}

First, we prove that $(T_G,\sigma_G)$ is a tree decomposition of $G$. From
the observation stated above, we have that for each vertex $v$ and edge $\{v,w\}$
there is a bag in $(T_G,\sigma_G)$ that contains $v$, respectively $\{v,w\}$.
For the third condition of tree decompositions, suppose $j_2$ is in $T_G$ on
the path from $j_1$ to $j_3$. If $v$ belongs to the bags of $j_1$ and $j_3$,
then all vertices in $H(v)$ belong in
$(T_H,\sigma_H)$ to the bags of $j_1$ and $j_3$, hence by the properties of
tree decompositions to the bag of $j_2$, and hence $v \in \sigma_G(j_2)$. It
follows that the preimage of each vertex in $V_G$ is a subtree of $T_G$. The
total weight of vertices in a bag in $(T_G,\sigma_G)$ is never larger than the
size of the corresponding bag in $(T_H,\sigma_H)$. This proves the first claim.

The above algorithm can also be reversed. If we take a tree decomposition
$(T_G,\sigma_G)$ of $G$, we can obtain one of $H$ by replacing in each bag
each vertex $v$ by the clique that results from blowing up $G$. The size of a
bag in the tree decomposition of $H$ now equals the total weight of the
vertices in $G$; hence the width of $(T_G,\sigma_G)$ equals the weighted width
of the obtained tree decomposition of $H$; it follows that  the weighted width
of $(T_G,\sigma_G)$ is equal to the width of $(T_H,\sigma_H)$ minus $1$.

Running the algorithm on an optimal tree decomposition of $H$ and the reverse
on an optimal weighted tree decomposition of $G$ shows that the the weighted
tree\-width of $G$ is at most (respectively, at least) $1$ plus the treewidth of
$H$, concluding our proof.
\end{proof}
%---------------------------------------------------------------------------------------

We are now ready to prove our main theorem for algorithms.
%------------------------------------------------------------------
\begin{theorem}\label{thm:alg_main}
Let $d\geq 2$, $\alpha>0$, and~$\eps>0$ be constants and let $\gamma$ be a weight function such that $1\leq \gamma(t) = O(t^{1-1/d-\eps})$. Then there exist constants $\kappa$ and $\Delta$ such that for any intersection graph $G=(V,E)$ of an (unknown) set of $n$ similarly-sized $\alpha$-fat objects in~$\Reals^d$, there is a $\kappa$-partition $\cP$ with the following properties:
(i) $G_\cP$ has maximum degree at most $\Delta$, and
(ii) $G_\cP$ has weighted treewidth $O(n^{1-1/d})$.
Moreover, such a partition $\cP$ and a corresponding tree decomposition of weight $O(n^{1-1/d})$ can be computed in $2^{O(n^{1-1/d})}$ time.
\end{theorem}

\begin{proof}
We use the greedy partition $\cP$ built around a maximal independent set (Definition~\ref{def:greedypart}). By Lemma~\ref{lem:gettreewidthbound}, the weighted treewidth of $G_\cP$ is $O(n^{1-1/d})$.
%Lemma~\ref{lem:getpartition} provides a partition $\cP$ built around a maximal independent set. By Lemma~\ref{lem:gettreewidthbound}, the weighted treewidth of $G_\cP$ is $O(n^{1-1/d})$.

To get a tree decomposition, consider the above partition again, with weight function $\gamma(t) = O(t^{1-1/d-\eps})$. We work on the contracted graph $G_\cP$; we intend to simulate the weight function by modifying $G_\cP$. Let $H$ be the graph we get from  $G_\cP$ by blowing up each vertex $v_C$ by an integer that is approximately the weight of the corresponding class, more precisely, we blow up $v_C$ by $\lceil\gamma(|C|)\rceil$. By Lemma~\ref{lem:weighttoclique}, its treewidth (plus one) is a $2$-approximation of the weighted treewidth of $G$ (since $\gamma(t)\ge 1$). Therefore, we can run a treewidth approximation algorithm that is single exponential in the treewidth of $H$. We can use the algorithm from either \cite{Robertson95} or \cite{Bodlaender16} for this, both have running time $2^{O(tw(H))}|V(H)|^{O(1)}=2^{O(n^{1-1/d})}(n\, \gamma(n))^{O(1)}=2^{O(n^{1-1/d})}$, and provide a tree decomposition whose width is a $c$-approximation of the treewidth of $H$. From this tree decomposition we gain a
tree decomposition whose weighted treewidth is a $2c$-approximation of the
weighted treewidth of $G_\cP$. In particular, we get a $2c$-approximation of the $\cP$-flattened treewidth of $G$.
This concludes the proof.
\end{proof}

% Using the above separator theorems, getting a $\kappa$-partitioned tree decomposition of width $O(n^{1-1/d})$ can be done in polynomial time by adapting~\cite{Reed97}\skb{is this a good citation?}; for this purpose, we only need that our separator theorems can be balanced with respect to any subset $S\subset V(G)$ instead of the complete vertex set of the graph.

\subsection{Basic algorithmic applications}

In this section, we give examples of how $\kappa$-partitions and weighted tree
decompositions can be used to obtain subexponential-time algorithms for
classical problems on geometric intersection graphs.

First, we make the following observation about tree decompositions.

\begin{observation}
Given a $\kappa$-partition $\cP$ and a weighted tree decomposition of $G_\cP$
of width $\tau$, there exists a nice tree decomposition of $G$
(i.e., a ``traditional'', non-partitioned tree decomposition) with the
property that each bag is a subset of the union of a number of partition
classes, such that the total weight of those classes is at most $\tau$.
\end{observation}

We can do this by creating a nice version of the weighted tree decomposition of
$G_\cP$, and then replacing every introduce/forget bag (that
introduces/forgets a class of the partition) by a series of introduce/forget
bags (that introduce/forget the individual vertices). We call such a
decomposition a \emph{traditional tree decomposition}.\index{tree decomposition!traditional \tilde} Using such a
decomposition, it becomes easy to give algorithms for problems for which we
already have dynamic-programming algorithms operating on nice tree
decompositions. We can re-use the algorithms for the leaf, introduce, join and
forget cases, and either show that the number of partial solutions remains
bounded (by exploiting the properties of the underlying $\kappa$-partition) or
show that we can discard some irrelevant partial solutions.

We present several applications for our framework, resulting in $2^{O(n^{1-1/d})}$ algorithms for various problems. In addition to the \textsc{Independent Set} algorithm for fat objects based on our separator, we also give a representation-agnostic algorithm for similarly sized fat objects. In contrast, the state of the art algorithm~\cite{MarxS14} requires the geometric representation as input. In the rest of the applications, our algorithms work on intersection graphs of $d$-dimensional similarly sized fat objects; this is usually a larger graph class than what has been studied. We have representation-dependent algorithms for \textsc{Hamiltonian Path} and \textsc{Hamiltonian Cycle}; this is a simple generalization from the algorithm for unit disks that has been known before~\cite{FominLPSZ17,ciac-hamiltonian}. For \textsc{Feedback Vertex Set}, we give a representation-agnostic algorithm with the same running time improvement, over a representation-dependent algorithm that works in $2$-dimensional unit disk graphs~\cite{FominLPSZ17}. For \textsc{$r$-Dominating Set},
we give a representation-agnostic algorithm for $d\ge 2$, which is the first subexponential algorithm in dimension $d\ge 3$, and the first representation-agnostic subexponential for $d=2$~\cite{MarxP15}. (The algorithm in~\cite{MarxP15} is for \textsc{Dominating Set} in unit disk graphs.) Finally, we give representation-agnostic algorithms for \textsc{Steiner Tree, $r$-Dominating Set, Connected Vertex Cover, Connected Feedback Vertex Set} and \textsc{Connected Dominating Set}, which are -- to our knowledge -- also the first subexponential algorithms in geometric intersection graphs for these problems.

In the following, we let $t$ refer to a node of the tree decomposition $T$,
let $\sig(t)$ denote the set of vertices in the bag associated with $t$, and let
$G[t]$ denote the subgraph of $G$ induced by the vertices appearing in bags in
the subtree of $T$ rooted at $t$. We fix our weight function to be
$\gamma(k)=\log(k+1)$.

\begin{theorem}\label{thm:independentset}
Let $\gamma(k)=\log(k+1)$. If a $\kappa$-partition and a weighted tree
decomposition of width at most $\tau$ is given, {\sc Independent Set} and {\sc
Vertex Cover} can be solved in time $2^{\kappa \tau} n^{O(1)}$.
\end{theorem}
\vspace{-1em}
\begin{proof}
A well-known algorithm (see, e.g., \cite{fptbook}) for solving
\textsc{Independent Set} on graphs of bounded treewidth, computes, for each
bag $\sigma(t)$ and subset $S\subseteq \sigma(t)$, the maximum size $c[t, S]$ of an
independent subset $\hat{S}\subset G[t]$ such that $\hat{S}\cap \sigma(t) = S$.

Let $t$ be a node of the weighted tree decomposition of $G_\cP$. Recall that the
corresponding bag is $\sigma(t)$, and it consists of partition classes. Let
$X_t=\bigcup_{C \in \sigma(t)} C$ be the set of vertices that occur in a
partition class in $\sigma(t)$. An independent set never contains more than
one vertex of a clique. Therefore, since from each partition class we can
select at most $\kappa$ vertices (one vertex from each clique), the number of
subsets $\hat{S}$ that need to be considered is at most  \[\prod_{C \in
\sigma(t)} (|C|+1)^\kappa = \exp\big(\sum_{C\in \sigma(t)} \kappa
\log{(|C|+1)}\big)=2^{\kappa\tau}.\]
This bound also holds for all the bags of the traditional tree decomposition that we create.
Applying the standard algorithm for \textsc{Independent Set} on the traditional
tree decomposition, using the fact that only solutions that select at most one
vertex from each clique get a non-zero value, we obtain the claimed algorithm for \IS.
\textsc{Vertex Cover} is simply the complement of \IS.
\end{proof}

Combining this result with Theorem \ref{thm:alg_main} gives the following result:

\begin{corollary}\label{cor:independentset}
Let $d\geq 2$ be a fixed constant, and let $G$ be an intersection graph of
similarly sized fat objects in $\Reals^d$. Then we can solve \IS and {\sc Vertex Cover}
on $G$ in $2^{O(n^{1-1/d})}$ time, even if the geometric representation
is not given.
\end{corollary}

In the remainder of this section, because we need additional assumptions that
are derived from the properties of intersection graphs, we state our results
in terms of algorithms operating directly on intersection graphs. However,
note that underlying each of these results is an algorithm operating on a
weighted tree decomposition of the contracted graph.

To obtain the algorithm for \IS, we exploited the fact that we can
select at most one vertex from each clique, and that thus, we can select at
most $\kappa$ vertices from each partition class. For \textsc{Dominating Set}, our
bound for the treewidth is however not enough. Instead, we need the following,
stronger result, which states that the weight of a bag in the decomposition
can still be bounded by $O(n^{1-1/d})$, even if we take the weight to be the
total weight of the classes in the bag \emph{and} of classes at most $r$
hops away in $G_\cP$.

\begin{theorem}\label{thm:smallnbh}
Let $d\geq 2, r\geq 1$ be constants and let $\gamma(t)=O(t^{1-1/d-\eps})$ be a weight function. Then there exists constants $\kappa,\Delta$, such that for any intersection graph $G$ of $n$ similarly-sized $d$-dimensional fat
objects there exists a $\kappa$-partition
$\cP$ and a corresponding $G_\cP$ of maximum degree at most $\Delta$, where
$G_\cP$ has a weighted tree decomposition with the additional property that
for any node $t$, the total weight of the partition classes
\[\{C\in\cP \mid \textrm{ there exist } v\in C,\, C'\in \sigma(t),\, v'\in C' \textrm{ such that $v$ and $v'$ are at distance at most $r$ in $G$ }\}\]
is $O(n^{1-1/d})$.
\end{theorem}

\begin{proof} As per Theorem~\ref{thm:alg_main}, there exist constants
$\kappa,\Delta=O(1)$ such that $G$ has a $\kappa$-partition $\cP$ in which
each class of the partition is adjacent to at most $\Delta$ other classes.

For any pair of classes $C,C' \in V(G_\cP)$ whose distance in $G_\cP$ is at
most $r$, we introduce a new copy of every object in $C$. This gives a new
intersection graph $G^r$. Note that for each object we now have at most
$1+\Delta+\Delta(\Delta-1)+ \dots \Delta(\Delta-1)^{r-1} = c =O(\Delta^r)$
copies. We create the following $\kappa c$-partition $\cP^r$: for each class
$C$ of the original partition, create a class that contains a copy of each
object of $C$ and a copy of each object from the classes at distance at most
$r$ from $C$. The resulting graph $G^r$ has at most $c n = O(n)$ vertices, and
it is an intersection graph of similarly-sized objects, and $G^r_{\cP^r}$ has constant-bounded maximum degree. Therefore, we can find
a weighted tree decomposition of $G^r_{\cP^r}$ of width $O(n^{1-1/d})$ by
using the machinery of  Section~\ref{sec:septw}.

This decomposition can also be used as a decomposition for $G$ and the
original $\kappa$-partition $\cP$, by replacing each partition class of
$\cP^r$ with the  original partition classes contained within; this increases
the width of the tree decomposition by at most a constant multiplicative
factor.
\end{proof}

\begin{theorem}\label{thm:dominatingset}
Let $d\geq 2$ and $r\geq 1$ be constants, and let $G$ be an intersection graph of
similarly sized fat objects in $\Reals^d$. Then {\sc Distance-$r$-Dominating Set} can be
solved on $G$ in $2^{O(n^{1-1/d})}$ time.
\end{theorem}

\begin{proof}
We first present the argument for \textsc{Dominating Set}. It is easy to see
that from each partition class, we need to select at most $\kappa^2(\Delta+1)$
vertices: each partition class can be partitioned into at most $\kappa$
cliques, and each of these cliques is adjacent to at most $\kappa(\Delta +1)$
other cliques. If we select at least $\kappa(\Delta+1)+1$ vertices from a
clique, we can instead select only one vertex from the clique, and select at
least one vertex from each neighboring clique.

We once again proceed by dynamic programming on a traditional tree
decomposition (see e.g. \cite{fptbook} for an algorithm solving \textsc{Dominating Set}
using tree decompositions). Rather than needing just two states per
vertex (in the solution or not), we need three: a vertex can be either in the
solution, not in the solution and not dominated, or not in the solution and
dominated. After processing each bag, we discard partial solutions that select
more than $\kappa^2(\Delta+1)$ vertices from any class of the partition. Note
that all vertices of each partition class are introduced before any are
forgotten, so we can guarantee that we indeed never select more than
$\kappa^2(\Delta+1)$ vertices from each partition class.

Whether a given vertex $v$ outside the solution is dominated or not is completely
determined by the vertices that are in its class and in neighboring classes. While the partial solution does not track this explicitly
for vertices that are forgotten, by using the fact that we need to select at
most $\kappa^2(\Delta+1)$ vertices from each class of the partition, and the fact
that Theorem~\ref{thm:smallnbh} bounds the total weight of the neighborhood
of the partition classes in a bag, we see that the number of partial solutions to be considered is at most
\[\prod_{C\in N(\sig(t))}
(|C_i|+1)^{\kappa^2(\Delta+1)} = \exp\left(\kappa^2(\Delta+1) \sum_{C\in N(\sig(t))} \log{(|C|
+ 1)}\right) = 2^{O(n^{1-1/d})},\]
where the product (resp., sum) is taken over $N(\sig(t))$, which denotes the set of
partition classes $C$ that appear in the current bag $\sig(t)$ or are a neighbors of
such a class. Therefore, we can solve \textsc{Dominating Set} in $2^{O(n^{1-1/d})}$ time.

For the generalization where $r>1$, the argument that we need to select at
most $\kappa^2(\Delta+1)$ vertices from each clique still holds: moving a
vertex from a clique with more than $\kappa(\Delta +1)$ vertices selected to
an adjacent clique only decreases the distance to any vertices it helps cover.
The dynamic programming algorithm needs, in a partial solution, to track at
what distance from a vertex in the solution each vertex is. This, once again,
is completely determined by the solution in partition classes at distance at
most $r$; the number of such cases we can bound using
Theorem~\ref{thm:smallnbh}.
\end{proof}

\subsection{Connectivity problems and the rank-based approach}

It is possible to combine our framework with the rank-based approach \cite{single-exponential} to obtain $2^{O(n^{1-1/d})}$-time algorithms for problems with connectivity constraints, avoiding the extra logarithmic factor associated with tracking connectivity constraints in the na\"ive way. To illustrate this, we now give an algorithm for Steiner Tree. We consider the following variant of Steiner Tree:

\defproblem{Steiner Tree}
{A graph $G=(V,E)$, a set of terminal vertices $K\subseteq V$ and integer $s$.}
{Decide if there is a vertex set $X\subseteq V$ of size at most $s$, such that $K\subseteq X$, and $X$ induces a connected subgraph of $G$.}

%We only consider the unweighted variant of Steiner Tree, as the weighted Steiner Tree problem has no $2^{o(n)}$ algorithm under ETH, even on a clique (so we should not expect Theorem~\ref{thm:steinertree} to hold for the weighted case).

The proof requires the following lemma.

\begin{lemma}\label{lem:stlem}
Let $\cP$ be a $\kappa$-partition of a graph $G$ where $G_\cP$ has maximum degree $\Delta$.
Let $A$ be a clique of a $\kappa$-sized clique partition of a class
$C\in \cP$. Suppose $X$ is a minimal solution for {\sc Steiner Tree} (i.e., no proper subset $X'$ of $X$ is also a solution). Then $X$ contains at most $\kappa
(\Delta+1)$ vertices from $A$ that are not also in $K$. Any
minimal solution thus contains at most $\kappa^2 (\Delta+1)$ vertices (that
are not also in $K$) from each partition class.
\end{lemma}

\begin{proof}
Consider the following process: To every vertex $v\in (A\cap X)\setminus K$ we greedily assign a
neighbor $u\in X\setminus A$ such that $u$ is adjacent to $v$ and $u$ is not
adjacent to any other previously assigned neighbor. We call such a neighbor a \emph{private neighbor}. We repeat this process until every vertex $v\in (A\cap X)\setminus K$ has been assigned a private neighbor.

If at some point it is not possible to assign a private neighbor to a vertex $v\in (A\cap X)\setminus K$, then $v$ has no neighbors in $X\setminus A$ that are not adjacent to a private neighbor of some vertex $v'\in (A\cap X)\setminus K$. Then either $v$ has no neighbors in $X$ outside $A$, or all such neighbors are already adjacent to some previously assigned private neighbor. In the former case all neighbors of $v$ are already adjacent (since they are in the clique $A$) and $X$ remains connected after removing $v$. In the latter case, $X$ also remains connected after removing $v$, since any neighbor (in $X$) of $v$ can be reached from some other vertex $v'\in (A\cap X)\setminus K$ through its private neighbor. Both cases contradict our assumption that $X$ is a minimal solution, so every vertex can be assigned a private neighbor.

We now note that since the neighborhood of $A$ can be covered by at most
$\kappa (\Delta +1)$ cliques, this gives us an upper bound on the number of
private neighbors that can be assigned and thus bounds the number of vertices
that can be selected from any partition class.
\end{proof}

\begin{theorem}\label{thm:steinertree}
Let $d\geq 2$ be a constant, and let $G$ be an intersection graph of similarly
sized fat objects in $\Reals^d$. Then {\sc Steiner Tree} can be solved on $G$ in
$2^{O(n^{1-1/d})}$ time.
\end{theorem}

\begin{proof}
The algorithm again works by dynamic programming on a traditional tree
decomposition. The leaf, introduce, join and forget cases can be handled as
they are in the conventional algorithm for \textsc{Steiner Tree} based on tree decompositions,
see e.g. \cite{single-exponential}. However, after processing
each bag, we can reduce the number of partial solutions that need to be
considered by exploiting the properties of the underlying $\kappa$-partition.

To this end, we first need a bound on the number of vertices that can be
selected from each class of the $\kappa$-partition $\cP$.

The algorithm for \textsc{Steiner Tree} presented in \cite{single-exponential}
is for the weighted case, but we can ignore the weights by setting them to
$1$. A partial solution is then represented by a subset $S\subseteq \sig(t)$
(representing the intersection of the partial solution with the vertices in
the bag $\sig(t)$), together with an equivalence relation on $S$ (which indicates
which vertices are in the same connected component of the partial solution).

By Lemma~\ref{lem:stlem} we select at most $\kappa^2 (\Delta+1)$ vertices from each partition
class, so we can discard partial solutions that select more than this number of
vertices from any partition class. Then the number of subsets $S$ considered
is at most

\[ \prod_{C\in\sigma(t)} (|C|+1)^{\kappa^2 (\Delta+1)} =
\exp\left(\kappa^2 (\Delta+1) \cdot \sum_{C\in\sigma(t)}
\log(|C|+1)\right)  \leq \exp\big(\kappa^2 (\Delta+1) n^{1-1/d}\big).
\]
For any such subset $S$, the number of possible equivalence relations is
$2^{\Theta(|S|\log{|S|})}$. However, the rank-based approach
\cite{single-exponential} provides an algorithm called ``\emph{reduce}'' that,
given a set of equivalence relations\footnote{What we refer to as
``equivalence relation'' is called a ``partition'' in \cite{single-exponential}.}
on $S$, outputs a representative set of equivalence
relations of size at most $2^{|S|}$. Thus, by running the reduce
algorithm after processing each bag, we can keep the number of equivalence
relations considered bounded by $2^{O(|S|)}$.

Since $|S| = O(\kappa^2 (\Delta+1) n^{1-1/d})$ (we select at most
$\kappa^2 (\Delta+1)$ vertices from each partition class and each bag contains
at most $O(n^{1-1/d})$ partition classes), for any subset $S$, the rank-based
approach guarantees that we need to consider at most $2^{O(\kappa^2 (\Delta+1)
{n^{1-1/d}})}$ representative equivalence classes of $S$.
\end{proof}

\begin{theorem}\label{thm:fvsalgo}
Let $d\geq 2$ be a constant, and let $G$ be an intersection graph of similarly
sized fat objects in $\Reals^d$.  Then {\sc Maximum Induced Forest}
(and \textsc{Feedback Vertex Set}) can be solved in
$2^{O(n^{1-1/d})}$ time in $G$.
\end{theorem}

\begin{proof}
We once again proceed by dynamic programming on a traditional tree decomposition corresponding to the weighted tree decomposition of $G_\cP$ of width $\tau$, where $\cP$ is a $\kappa$-partition, and the maximum degree of $G_\cP$ is at most $\Delta$. We describe the algorithm for
{\sc Maximum Induced Forest}, which also gives an algorithm for {\sc Feedback Vertex Set} as it is simply its complement.

Using the rank-based approach with \textsc{Maximum Induced Forest} requires some modifications to the problem, since the rank-based approach is designed to get maximum connectivity, whereas in \textsc{Maximum Induced Forest}, we aim to ``minimize'' connectivity (i.e., avoid creating cycles). To overcome this issue, the authors of \cite{single-exponential} add a special universal vertex $v_0$ to the graph (increasing the width of the decomposition by $1$) and ask (to decide if a Maximum Induced Forest of size $k$ exists in the graph) whether we can delete some of the edges incident to $v_0$ such that there exists an induced, connected subgraph, including $v_0$, of size $k+1$ in the modified graph that has exactly $k$ edges. Essentially, the universal vertex allows us to arbitrarily glue together the trees of an induced forest into a single (connected) tree. This thus reformulates the problem such that we now aim to find a connected solution.

The main observation that allows us to use our framework, is that from each clique we can select at most $2$ vertices (otherwise, the solution would become cyclic), and that thus, we only need to consider partial solutions that select at most $2\kappa$ vertices from each partition class. The number of such subsets is at most $2^{O(\kappa n^{1-1/d})}$. Since we only need to track connectivity among these $2\kappa$ vertices (plus the universal vertex), the rank-based approach allows us to keep the number of equivalence relations considered single-exponential in $\kappa n^{1-1/d}$. Thus, we obtain a $2^{O(\kappa n^{1-1/d})}n^O(1)$-time algorithm.
\end{proof}

\paragraph*{Additional Problems}

Our approach gives $2^{O(n^{1-1/d})}$-time algorithms on geometric intersection graphs of $d$-dimensional similarly-sized fat objects for almost any problem with the property that the solution (or the complement thereof) can only contain a constant (possibly depending on the ``degree'' of the cliques) number of vertices of any clique. We can also use our approach for variations of the following problems, that require the solution to be connected:

\begin{itemize}
    \item {\sc Connected Vertex Cover} and {\sc Connected Dominating Set}: these problems may be solved similarly to their normal variants (which do not require the solution to be connected), using the rank-based approach to keep the number of equivalence classes considered single-exponential. In case of {\sc Connected Vertex Cover}, the complement is an independent set, therefore the complement may contain at most one vertex from each clique. In case of {\sc Connected Dominating Set}, it can be shown that each clique can contain at most $O(\kappa^2\Delta)$ vertices from a minimum connected dominating set.

    \item  {\sc Connected Feedback Vertex Set}: the algorithm for Maximum Induced Forest can be modified to track that the complement of the solution is connected, and this can be done using the same connectivity-tracking equivalence relation that keeps the solution cycle-free.

\end{itemize}

\begin{theorem}
For any constant dimension $d\geq 2$, {\sc Connected Vertex Cover}, {\sc Connected Dominating Set} and {\sc Connected Feedback Vertex Set} can be solved in time $2^{O(n^{1-1/d})}$ on intersection graphs of similarly-sized $d$-dimensional fat objects.
\end{theorem}

\paragraph*{Hamiltonian Cycle} Our separator theorems imply that
\textsc{Hamiltonian Cycle} and \textsc{Hamiltonian Path} can be solved in
$2^{O(n^{1-1/d})}$ time on intersection graphs of sim\-i\-larly-sized
$d$-dimensional fat objects. However, in contrast to our other results, this
requires that a geometric representation of the graph is given. Given a
geometric representation, we can compute a
$1$-partition $\cP$ where $G_\cP$ has constant degree. It can be shown
that a cycle/path only needs to use at most two edges between each pair of
cliques; see e.g. \cite{ito-hamiltonian, ciac-hamiltonian} and that we can
obtain an equivalent instance with all but a constant number of vertices
removed from each clique. Our separator theorem implies that this graph has
treewidth $O(n^{1-1/d})$, and \textsc{Hamiltonian Cycle} and
\textsc{Hamiltonian Path} can then be solved using dynamic programming on a
tree decomposition.

\begin{theorem}
For any constant dimension $d\geq 2$, {\sc Hamiltonian Cycle} and {\sc
Hamiltonian Path} can be solved in time $2^{O(n^{1-1/d})}$ on the intersection
graph of similarly-sized $d$-dimensional fat objects which are given as input.
\end{theorem}

Note that the geometric representation is only needed to
ensure that we can find a $1$-partition $\cP$ such that $G_\cP$ has constant degree. Without the geometric representation the complexity
of computing such a $1$-partition is unknown, and a challenging open question.

%% file: lower_framework.tex
\section{The lower-bound framework}\label{sec:lowerbounds}

The goal of this section is to provide a general framework to exclude algorithms with running time $2^{o(n^{1-1/d})}$ in intersection graphs.
To get the strongest results, we show our lower bounds where possible for a more restricted graph class, namely subgraphs of $d$-dimensional induced grid graphs. Induced grid graphs are intersection graphs of unit balls, so they are a subclass of intersection graphs of similarly sized fat objects.
We need to use a different approach for $d=2$ than for $d>2$; this is because of
the topological restrictions introduced by planarity. Luckily, the difference between $d=2$ and $d>2$ is only in the need of two
different ``embedding theorems''; when applying the framework to specific problems, the same gadgetry works both for $d=2$ and for $d>2$. In particular, in $\Reals^2$, constructing crossover gadgets is not necessary with our framework. To apply our framework, we need a graph problem $\cP$ on grid graphs in $\Reals^d$, $d\geq 2$. Suppose that $\cP$ admits a reduction from $3$-SAT using constant size variable and clause gadgets and a wire gadget, whose size is a constant multiple of its length. Then the framework implies that $\cP$ has no $2^{o(n^{1-1/d})}$ time algorithm in $d$-dimensional grid graphs for all $d\ge 2$, unless ETH fails. We remark that such gadgets can often be obtained by looking at classical NP-hardness proofs in the literature, and introducing minor tweaks if necessary.

\subsection{Lower bounds in two dimensions}

To prove lower bounds in two-dimensional grids, we introduce an intermediate
problem.

For an integer $n$, let $[n]=\{1,\dots, n\}$. We denote by $G^2(n_1,n_2)$ the two-dimensional grid graph with vertex set
$\civ{n_1} \times \civ{n_2}$. We say that a graph $H$ is \emph{embeddable} in
$G^2(n_1,n_2)$ if it is a topological minor of $G^2(n_1,n_2)$, i.e., if $H$\index{minor!topological \tilde}
has a subdivision that is a subgraph of $G^2(n_1,n_2)$. Finally, for a given
3-CNF formula $\phi$, its \emph{incidence graph} $G_{\phi}$ is the bipartite graph on\index{graph!incidence \tilde}
its variables and clauses, where a variable vertex and a clause vertex are
connected by an edge if the variable appears in the clause.

We say that a CNF formula $\phi$ is a $(3,3)$-CNF formula if all clauses in $\phi$ have size at most $3$ and each variable
occurs at most $3$ times.\footnote{Crucially, we allow clauses of size $2$, as formulas with clause size exactly $3$ and at most $3$ occurrences per variable are trivially satisfiable~\cite{DBLP:journals/dam/Tovey84}.} Note that in such
formulas the number of clauses and variables is within constant factor of
each other. The $(3,3)$-SAT problem asks to decide the satisfiability of
a $(3,3)$-CNF formula.

\begin{proposition}\label{prop:33SAT}
There is no $2^{o(n)}$ algorithm for $(3,3)$-SAT unless ETH fails.
\end{proposition}

\begin{proof}
By the sparsification lemma of Impagliazzo, Paturi and Zane
\cite{ImpagliazzoPZ01}, satisfiability on 3-CNF formulas with $n$ variables and
$\Theta(n)$ clauses has no $2^{o(n)}$ algorithm under the ETH. Let $\phi$ be
such a formula. If a variable $v$ occurs $k>3$ times in $\phi$, then
we can replace $v$ with a new variable at each occurrence. Call these new
variables $v_i\, (i=1,\dots, k)$. Now, add the following clauses to the
formula:

\begin{equation}\label{eq:makeequal}
(v_1 \vee \neg v_2)\wedge(v_2 \vee \neg v_3)\wedge \dots \wedge(v_{k-1} \vee \neg v_k)\wedge(v_k \vee \neg v_1).
\end{equation}

It is easy to see that the resulting formula is a $(3,3)$-CNF formula of
$O(n)$ variables and clauses, and it can be created in polynomial time from
the initial formula. Next, we argue that the new formula is satisfiable if and only if $\phi$ is satisfiable.
If $\phi$ is satisfiable, then the new formula is also satisfied by the assignment where for each $v$ we set all  the variables $v_i$ to be equal to $v$. If the new formula is satisfiable, then for each $v$ the added clauses can only be satisfied if $v_1=v_2=\dots=v_k$; therefore we can set $v=v_1$ and such an assignment will satisfy all the original clauses.
Consequently, a $2^{o(n)}$ algorithm for $(3,3)$-SAT would also give a $2^{o(n)}$ algorithm to evaluate the satisfiability of $\phi$, which contradicts ETH.
\end{proof}

Our intermediate problem, {\sc Grid Embedded SAT}, asks to determine the
satisfiability of a $(3,3)$-CNF formula whose incidence graph is embedded
in an $n\times n$ grid:

\defproblem{Grid Embedded SAT}{A $(3,3)$-CNF formula $\phi\,$ together
with an embedding of its incidence graph $G_{\phi}$ in $G^2(n,n)$.}%
{Is there a satisfying assignment?}

\begin{theorem}\label{thm:GridEmbSATlower}
\textsc{Grid Embedded SAT} has no $2^{o(n)}$ algorithm under ETH.
\end{theorem}

\begin{proof}
Consider a $(3,3)$-CNF formula $\phi$. As a first step, we generate a grid drawing $\cD_{\phi}$ in $\Reals^2$
of the incidence graph of $\phi$ the following way. Assign the point $(3i,0)$ to vertex $v_i$ of $G_\phi$, as depicted in
Figure~\ref{fig:griddrawing}, and add horizontal grid segments to its left (resp. left and right) to
vertices of degree $2$ and $3$; this way each vertex in $G_{\phi}$ of degree
$k$ is assigned to a \emph{group} of $k$ consecutive grid points. Finally,
for each edge of $G_{\phi}$, we add two vertical segments and a horizontal
segment in a $\sqcap$ shape that connects two points corresponding to the
group of its endpoints; the height of this segment for edge $j$ is set to $j$. This drawing assigns a unique grid point to each
vertex of $G_{\phi}$ and a grid path to each edge (with intersections).

\begin{figure}[t]
\begin{center}
\includegraphics{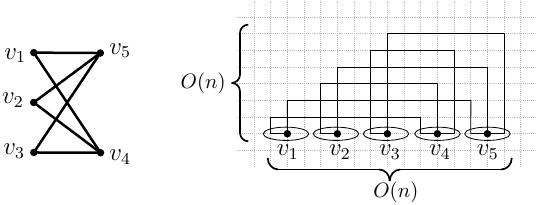}
\end{center}
\caption{Left: the incidence graph of the formula $(x_1 \vee \neg x_2 \vee x_3) \wedge (\neg x_1 \vee  x_2 \vee x_3 )$. Right: a drawing $\cD_{\phi}$ of the incidence graph on the grid.
Each edge is a grid path, the main part of which is a $\sqcap$ shape with a
unique integer height.}%
\label{fig:griddrawing}
\end{figure}

Next, we need to planarize $\phi$. To this end, we use a modified
version of the crossover gadget from Lichtenstein's classical planar 3-SAT
reduction~\cite{Lichtenstein82}. The cross\-over gadget is built on a constant size
$3$-CNF formula with two pairs of special variables, $a,a'$ and $b,b'$. The incidence graph of the formula is planar and has maximum degree six, see Figure~\ref{fig:licht}. The
incidence graph has a planar embedding where the vertices corresponding to the
special variables occur in the order $a,b,a',b'$ around the unbounded face; furthermore, the degree of $a$ and $b$ within this graph is two. The added variables and clauses of the gadget ensure that in all satisfying assignments $a=a'$ and $b=b'$.
One uses the gadget the following way: suppose two edges, starting at variable vertices $u$ and $v$ are crossing in our drawing. Then we identify $u$ and $a$, and also $v$ and $b$. On the other side of the crossing, we add the new variable vertices $a'$ and $b'$; finally, the crossing itself gets replaced by the above mentioned embedding of the gadget.

\begin{figure}[t]
\begin{center}
\includegraphics[scale=0.7]{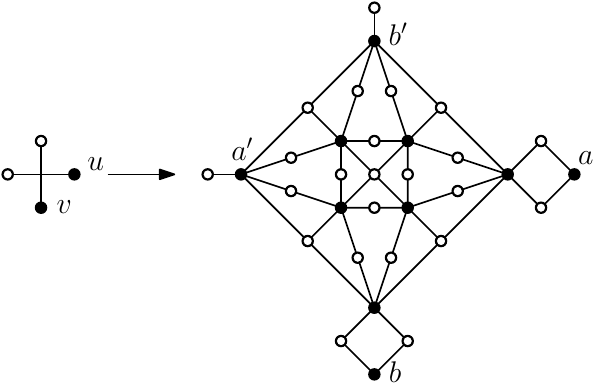}
\end{center}
\caption{Lichtenstein's crossing gadget.}%
\label{fig:licht}
\end{figure}

Note that this modification can at most double the degree of each original variable vertex, and the maximum degree of newly introduced variable vertices is six. Therefore, the resulting formula $\phi'$ is planar, and it has degree at most six. For these high degree vertices we introduce a degree-decreasing gadget,
depicted in Figure~\ref{fig:degdec}, where the new formulas (unlabeled
vertices in the figure) are the same as in~\eqref{eq:makeequal} for $k=6$.
After introducing these new variables and clauses, we get a planar $(3,3)$-CNF formula $\phi''$.
Due to the nature of our modifications the formula $\phi''$ is satisfiable if and only if $\phi$ is satisfiable.

\begin{figure}[t]
\begin{center}
\includegraphics[scale=0.7]{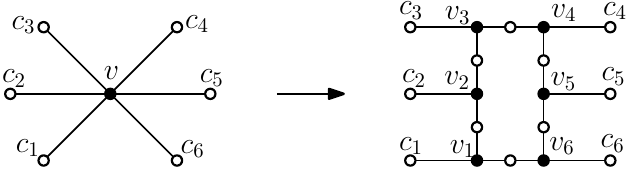}
\end{center}
\caption{A low-degree gadget to help grid embeddings.}
\label{fig:degdec}
\end{figure}

Using the earlier drawing $\cD_{\phi}$ we can easily produce a grid drawing
for $G_{\phi''}$. Clearly there is a constant $c_1$ such that one can draw our
degree decreasing gadget of size six inside, and there is a constant $c_2$
such that one can draw the crossing gadget (with our degree decreasing gadget
added where necessary) in a grid of size $c_2\times c_2$. By switching the
grid underneath to a grid with $2\max(c_1,c_2)$ times the density, we can
introduce these gadgets at the crossings, and also add the degree-decreasing
gadget in place of new high degree variable vertices. This results in a grid
drawing of $G_{\phi''}$, where the grid has size $O(n)\times O(n)$, and
therefore the whole construction has $O(n^2)$ vertices. The algorithm to
obtain this embedding runs in polynomial time. This completes our reduction.
\end{proof}

Note that \textsc{Grid Embedded SAT} is solvable in $2^{O(n)}$ time, since it
reduces to \textsc{Planar SAT} on $O(n^2)$ variables and clauses, which in
turn has an algorithm with running time $2^{O(\sqrt{t})}$ where $t$ is the
number of variables and clauses~(see e.g., \cite{Woeginger2003}).

\subsection{Lower bounds in higher dimensions -- Cube Wiring}

% Achieving a tight lower bound for $d\ge 3$ depends on  how small of a grid we
% can embed the incidence graph of a $(3,3)$-SAT formula into. To formally state
% our result, we need to introduce the following notation:

%For an integer $n$, let $[n]=\{1,\dots, n\}$. 
For a vector $\bn:=
(n_1,\dots,n_d)$ in $\Ints_+^d$, let $\mybox_d(\bn)= [n_1]\times \dots\times
[n_d]$. Let $G^d(\bn)$ be the graph whose vertex set $V(G)$ is
$\mybox_d(\bn)$, and where $\bx,\by\in V(G)$ are connected if and only if they
are at distance $1$ in $\Reals^d$. The integer points of $\Reals^d$ can be
divided into parallel \emph{layers}. The layer at ``height'' $h\in \Ints$ is
defined as $\ell(h)=\{ \bx \in \Ints^d \;|\; x_d = h \}$. Let $\emb^h:
\Reals^{d-1} \rightarrow \Reals^d$ be the function that maps $\Reals^{d-1}$
into $\ell(h)$ as follows: $\emb^h(x_1,\ldots,x_{d-1}) =
(x_1,\ldots,x_{d-1},h)$.

In what follows, $\bn$ denotes a $(d-1)$-dimensional vector, and we will
often denote $d$-dimensional vectors as a concatenation of a $d-1$ and a
$1$-dimensional vector. For example, $G^d(\bn,h)$ denotes the grid graph for
the vertex set $[n_1]\times\dots\times[n_{d-1}]\times[h]$.

Let $P,Q$ be equal-size subsets of $\mybox(\bn)$. Let $M$ be a perfect
matching of the graph $G_{P\times Q} := (P\cup Q,P\times Q)$. We say that
\emph{$M$ can be wired in $G^d(c \mathbf{n}, h)$} where $c$ and $h$ are
positive integers, if there are vertex-disjoint paths in
$G^d(c \mathbf{n}, h)$ that connect $\emb^1(\bp)$ to $\emb^h(\bq)$ for all
$(\bp,\bq)\in M$. Note that $G^d(c \mathbf{n}, h)$ consists of $h$ layers,
each of which is a copy of $\mybox(c \bn)$ at a different height. 

We will refer to the embedding in $\Reals^d$ of the path representing a pair $(\mathbf{p},\mathbf{q})$ as a \emph{wire}.

Let $k$ be a positive integer, and consider a point set $P\subseteq \Ints^{d-1}$. The set $P$ is \emph{$k$-spaced} if there is an integer $0\leq r < k$ such that for any $x=(x_1,\dots,x_{d-1})\in P$ we have $x_i \equiv r \mod{k}$ for all $i=1,\dots,d-1$.

\begin{theorem}\label{thm:boxwiring} \textbf{(Cube Wiring Theorem)}
Let $d\ge 3$, $\bn\in \Ints_+^{d-1}$, and let $P$ and $Q$ be two equal-size $2$-spaced
subsets of $\mybox_{d-1}(2\bn)$. Let $M$ be a perfect matching in $G_{P\times
Q} =(P\cup Q,P\times Q)$. Then $M$ can be wired in $G^d(2\bn,h)$, where
$h=O(\sum_{i=1}^{d-1} n_i)$.
\end{theorem}

Our original (and qualitatively different) proof for a slightly different formulation can be found in Kisfaludi-Bak's thesis~\cite{doktori}.
As it turns out, we can also derive the above theorem from the following result of Thompson and Kung~\cite{DBLP:journals/cacm/ThompsonK77} that concerns sorting on parallel processors that are connected into a $d-1$-dimensional grid cube. In their setting, a set of processors are connected in a $d-1$-dimensional hypercube. In a time step, some disjoint pairs of neighboring processors can exchange their content. In the \emph{sorting problem}, each processor receives a number as input, and the goal is to sort the numbers so that their ordering coincides with the lexicographical order of the processors using a small number of time steps.  Note that the time corresponds to the last axis in our setting. Their main result can be stated the following way.

\begin{theorem}\textbf{(Thompson and Kung~\cite{DBLP:journals/cacm/ThompsonK77}, 1977, paraphrased)} 
The sorting problem on processors connected in a grid $\bn$ can be solved in $O(\sum_{i=1}^{d-1} n_i)$ time.
\end{theorem}

\begin{figure}
\centering
\includegraphics{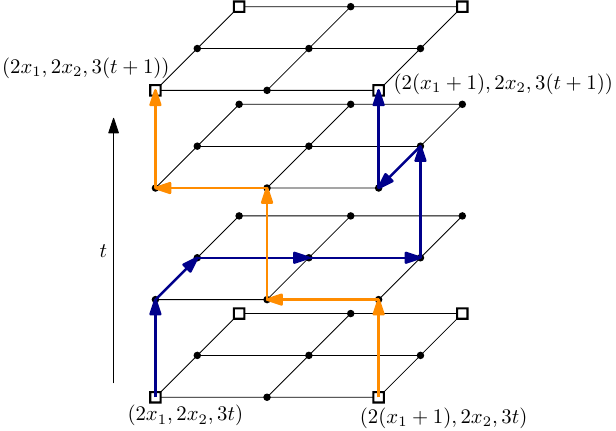}
\caption{Exchanging two neighboring wires in a $2\times 2 \times 3$ grid box.}\label{fig:wire_exchange}
\end{figure}

To prove the Cube Wiring Theorem, one just needs to define a wire for each number that tracks its location in the grid through the procedure. For a number that is at position $(x_1,\dots,x_{d-1})$ at time $t$, we associate the point $2x_1,2x_2,\dots,2x_{d-1},3t$, and ensure that it is on its wire. It is sufficient to show that the exchange of neighboring numbers can be done in a $2 \times 2 \times \dots 2 \times 3$ grid cube. In fact, it is sufficient to show that the exchange is possible in a $2\times 2 \times 3$ grid box, as the same method can be used in higher dimensions. For exchanging elements $(2x_1,2x_2,3t)$ and $(2(x_1+1),2x_2,3t)$, we need to get the wire of the first element to $(2(x_1+1),2x_2,3(t+1))$ and of the second to $(2x_1,2x_2,3(t+1))$. See Figure~\ref{fig:wire_exchange} for an illustration of such an exchange. The wire of the first follows the path
\begin{multline*}
(2x_1,2x_2,3t), (2x_1,2x_2,3t+1), (2x_1,2x_2+1,3t+1), (2x_1+1,2x_2+1,3t+1),\\ (2(x_1+1),2x_2+1,3t+1), (2(x_1+1),2x_2+1,3t+2), (2(x_1+1),2x_2,3t+2),(2(x_1+1),2x_2,3(t+1)).
\end{multline*}
For the second, we use the wire
\begin{multline*}
(2(x_1+1),2x_2,3t), (2(x_1+1),2x_2,3t+1), (2x_1+1,2x_2,3t+1),\\ (2x_1+1,2x_2,3t+2), (2x_1,2x_2,3t+2)  , (2x_1,2x_2,3(t+1)).
\end{multline*}

A rotation of this wiring is used for a neighboring wire pair $(2x_1,2x_2,3t)$ and $(2x_1,2(x_2+1),3t)$. Notice that the created wires remain vertex disjoint. This concludes the proof of Theorem~\ref{thm:boxwiring}.

\subsection{Strengthening the Cube Wiring Theorem}

We now turn to giving a stronger version of Cube Wiring. Note that this stronger version is not necessary for the lower bounds to hold, but it may be necessary in other applications.

\begin{theorem}\label{thm:strongwiring} \textbf{(Strong Cube Wiring Theorem)}
There exists a constant $c$ such that for any $d\ge 3$ the following hold.
Let $\bn\in \Ints_+^{d-1}$, and let $P$ and $Q$ be two equal-size $2$-spaced
subsets of $\mybox_{d-1}(2\bn)$. Let $M$ be a perfect matching in $G_{P\times
Q} =(P\cup Q,P\times Q)$. Then $M$ can be wired in $G^d(c\bn,cn_{\max})$, where
$n_{\max}$ is the largest coordinate of $\bn$.
\end{theorem}

The proof is based on the following lemma.

\begin{figure}[t]
\begin{center}
\includegraphics[width = 0.75\textwidth]{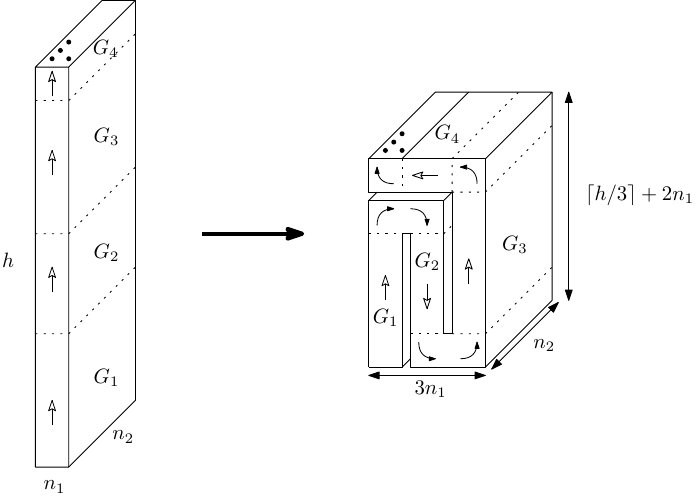}
\end{center}
\caption{Decreasing the height of a wiring by adding elbows. We decompose the
wiring into four parts, and add elbows as required (denoted by bent arrows).} \label{fig:squiggle}
\end{figure}

\begin{lemma}\label{lem:squiggle}
If a matching can be wired in $G^d((n_1,\dots,n_{d-1}),h)$, then it can
also be wired in\\ $G^d((3n_1,n_2,\dots,n_{d-1}),\ceil{h/3}+2n_1)$.
\end{lemma}

\begin{proof} The idea is to add artificial turns into the wiring, and create
the snake-like structure seen in Figure~\ref{fig:squiggle} that has roughly
one-third the height of the original wiring. The wiring in three dimensions is
created as depicted in the figure. Note that for each wire point we only
introduce changes along the first and last coordinate $x_1$ and $x_d$, and all
other wire points retain their original coordinates $x_2,\dots,x_{d-1}$.

Let $W$ be some set of wire points in the starting layer $\ell(1)$ . We
introduce elbows that bend around the $x_1$ axis, that is, it brings the
wires from the ``bottom'' face of a box to an adjacent face. More precisely,
for a wire $w$ with wire point $(x_1,\dots,x_{d-1},1)$ in $\ell(1)$, we first
increase the last coordinate until we reach $x'_d=n_1-x_1+1$, then increase
along the first coordinate until we reach $x'_1=n_1+1$. Doing this for all the
wires in $W$ creates a right elbow. One can also create a left elbow
analogously. The elbow keeps the wires disjoint and preserves ordering: the
wire points at the beginning of the elbow can be mapped to to the wire points
at the end of the elbow by a rotation of angle $\pi/2$.

We can build a wiring based on the original wiring $G$ the
following way. Decompose $G$ into four distinct wirings, by cutting the wires
at their wire points of height $\ceil{h/3}$, $\ceil{2h/3-n_1}$ and $h-n_1$.
Let $G_1,G_2,G_3$ and $G_4$ be the resulting partial wirings. To construct the
new wiring, we start with $G_1$ and keep it unchanged. Then we introduce two
right elbows in succession, and attach to that $G_2$, which is rotated in the
$x_1x_d$ plane by $\pi$. Next, we attach two left elbows in succession, and
add a translate of $G_3$. Next, we add a left elbow, and the wiring $G_4$ that
is rotated by $\pi/2$ in the $x_1x_d$ plane. Notice that the resulting box has
the required size: the new height is $\ceil{h/3}+2n_1$, and the width along the
first coordinate has tripled from $n_1$ to $3n_1$; the width of the bounding
box is the same along all other axes. Finally, notice that the wire points at
the end of the wiring are in the same position of the box as they were in the
original wiring.
\end{proof}

We are now ready to prove the Strong Cube Wiring Theorem.

\begin{proof}[Proof of Theorem~\ref{thm:strongwiring}]
Let $G$ be the wiring that is given by the Cube Wiring Theorem. The wiring
has base $\mybox(2n_1,\dots,2n_{d-1})$ and height $h<c'dn_{\max}$ for some
constant $c'$. Without loss of generality, assume that $n_{\max}=n_1 \geq
n_2\geq\dots\geq n_{d-1}$ and that each $n_i$ and $h$ are powers of three.
This can be achieved by rounding all of them up to the nearest power of three so that $G$ has base  $\mybox(6n_1,\dots,6n_{d-1})$.

We now apply Lemma~\ref{lem:squiggle} to each of the coordinates
$x_1,\dots,x_{d-1}$ in succession. The resulting box has base
$\mybox(18n_1,\dots,18n_{d-1})$ and its height is 
\[
h'=\frac{h}{3^{d-1}}+\frac{2n_1}{3^{d-2}}+\frac{2n_2}{3^{d-3}}
+\dots+\frac{2n_{d-1}}{3^{0}}
\leq \frac{h}{3^{d-1}} +
2n_{\max}\sum_{j=0}^{d-2}\frac{1}{3^j}
< \frac{h}{3^{d-1}}
+ 3n_{\max}.\]
Notice that for all $d\geq 3$ we have that $3^{d-1}>d$,
therefore 
\[h'<\frac{h}{3^{d-1}} +
3n_{\max} < \frac{c'dn_{\max}}{3^{d-1}} + 3n_{\max}<\left(c'+3\right)n_{\max}.\]
Picking $c=\max\left(18,\left(c'+3\right)\right)$ concludes the proof.
\end{proof}

\begin{remark}
The $2$-spaced requirement in both the original and the strong version of Cube Wiring was added for technical reasons: one can rearrange the point sets $P$ and $Q$ to be tightly packed in $\emb^1(\bn)$ and $\emb^h(\bn)$ respectively~\cite[Lemma 7.6]{doktori}.
\end{remark}

\subsection{Minors in grids}

We describe a corollary of Strong Cube Wiring in terms of graph minors. Recall that a
graph $G$ is a minor of a graph $G'$ if there is way to get from $G'$ to $G$ by\index{minor}
some sequence of edge contractions, edge deletions and vertex deletions. While
the theorem itself is not used elsewhere in the paper, we think that it is
interesting in its own right.

\begin{theorem}\label{thm:cubewiringminor}
There exists a constant $c$ such that for all $d\ge 3$, any graph with $m$ edges and no isolated vertices
is the minor of the $d$-dimensional grid hypercube of side length
$cm^{\frac{1}{d-1}}$.
\end{theorem}

\begin{proof}
Let $G$ be an arbitrary graph with $m$ edges. We expand all vertices $v$ of
$G$ into a path $P_v$ of length $\deg_G (v)$, that is, replace $v$ with a path $P_v$ of length $\deg(v)$,
where each vertex of $P_v$ is adjacent to a single neighbor of $v$. We also
subdivide each original edge $e=uv$ of $G$ by two new vertices, $w_{eu}$
(adjacent to $u$) and $w_{ev}$ (adjacent to $v$); let $G'$ be the graph that
we end up with after these modifications. Let $P\eqdef \bigcup_{v\in V} V(P_v)$ and
 $Q \eqdef \{w_{ev}\;|\;e\in E(G)\text{ and } v\in e\}$.

There is a constant $c'$ such that we can find an isomorphism $\phi_P$ from
$G'[P]$ into  a subgraph of $\Gamma = G_{d-1}((c'm^{1/(d-1)},\dots,c'm^{1/(d-1)}),c'm^{1/(d-1)})$, and
similarly an isomorphism $\phi_Q$ from $G'[Q]$ into a subgraph of
$\Gamma$. By the Strong Cube Wiring
Theorem, there are vertex-disjoint paths connecting $\phi_P(P)$ to $\phi_Q(Q)$
in $cm^{1/(d-1)}$ layers with the matching $M=E(G')\cap (P \times Q)$, where
$c$ is a constant. Since $E(G') = G'[P] \cup G'[Q] \cup M$, this shows that
$G'$ is a minor of the grid hypercube. We also have that $G$ is a minor of
$G'$, therefore $G$ is a minor of the grid hypercube of side length
$cm^{\frac{1}{d-1}}$.
\end{proof}

We could also consider \emph{topological minors} instead. We would like\index{minor!topological \tilde}
to find an edge subdivision of some given graph $G$ in a grid cube. In this
case we are forced to bound the degree of $G$ with $2d$, since the maximum
degree of the grid is an upper bound on the maximum degree of its topological
minors. As in the proof above, we need to use a grid of side length at least $c'n^\frac{1}{d-1}$ for some $(2d)^{1/d}<c'=O(1)$ in order to fit the starting points of the at most $2dn$ wires into a single layer. Then invoking the Strong Cube Wiring yields the following theorem.

\begin{theorem}\label{thm:cubewiringtopminor}
There exists a constant $c$ such that for all $d\ge 3$, any graph with $n$
vertices and maximum degree $2d$ is the topological minor of the
$d$-dimensional grid hypercube of side length $cn^{\frac{1}{d-1}}$.
\end{theorem}

\section{Applying the lower-bound framework}

In order to construct reductions for our problems, we can often reuse gadgetry
from classical \NP-completeness proofs. The simplest approach would be to start with a problem on planar graphs, and try to create a grid graph based on that. Unfortunately, this approach is not sufficient for ETH-tight lower bounds for the following reason. Drawing a planar graph (even of
maximum degree $3$) on $n$ vertices may require an $\Theta(n)\times \Theta(n)$
grid. Usually, the ETH-based lower bound for the starting planar problem is of the form
$2^{\Omega(\sqrt{n})}$. Trying to realize the planar graph as a grid graph results in a
graph of size $n^2$, since connecting distant vertices requires us to
subdivide the edge with $\Omega(n)$ grid points. Therefore the resulting lower
bound would be of the form $2^{\Omega(n^{1/4})}$, which is not ETH-tight.

Instead, it is usually required to maintain a grid drawing carefully in a grid
of size $O(n)\times O(n)$. In our reductions, we will either start with {\sc
Grid Embedded SAT}, or with an arbitrary (3,3)-CNF
formula and a specific grid drawing of its incidence graph (with crossings),
similar to what we have done in the proof of Lemma~\ref{thm:GridEmbSATlower}.

Recall that our goal is to prove the lower bounds in the most restricted graph
class possible. Thus, where possible we aim to get a lower bound in induced grid graphs.
There are two cases where we do not succeed in obtaining such a lower bound.

\begin{enumerate}
\item \textsc{Independent Set} and \textsc{Vertex Cover} are
solvable in polynomial time on bipartite graphs, because they are equivalent to
matching~\cite{Konig1931}, and therefore can be found using a bipartite
matching algorithm~\cite{HopcroftK73}. Since $d$-dimensional grid graphs are
bipartite, the lower bounds can only be achieved in some larger graph class.
Hence, for \IS and \textsc{Vertex Cover} we will prove our lower bounds for
unit ball graphs. Regardless, the general strategy remains the same; we
can use the same type of gadgetry and realize the constructed wires by
mimicking the grid-embedded drawing or Cube Wiring.
\item In order to prove
lower bounds in $2$-dimensional grid graphs, one needs to be able to create
gadgetry that has maximum degree $4$, and other desirable properties. Such
gadgets are often not known in the literature, and in case of
\textsc{Connected Vertex Cover} and \textsc{Connected Feedback Vertex Set}, we
are not able to create such gadgets. Instead, we end up proving the lower
bounds for unit disk graphs, while still having the lower bound for
$d$-dimensional grid graphs when $d\geq 3$. In case of
\textsc{Connected Vertex Cover}, we are also able to prove the lower bound
for (non-induced) grid graphs in two dimensions.
\end{enumerate}

A key step in many of these reductions is refinement. A $k$-\emph{refinement}\index{refinement (of a grid)}
of a drawing $\cD \subset \Reals^d$ inside a grid is obtained by scaling the
drawing by a factor of $k$. This means that an axis-parallel grid segment in the
drawing becomes an axis-parallel grid segment whose length is a multiple of
$k$. If $\cD$ is a drawing of a grid graph, then by subdividing each segment
of the $k$-refinement using $k-1$ inner grid points, we get an induced grid
graph. If we say that a drawing or a grid is refined without specifying $k$,
then it means that we introduce some refinement for some large enough
constant $k\in \Nats_+$ that is chosen so that certain conditions hold.

\subsection{\textsc{Dominating Set}}\label{sec:dslower}

We prove the following lower bound for \textsc{Dominating Set}.

\begin{theorem}
Let $d\geq 2$ be a fixed constant. Then there is no $2^{o(n^{1-1/d})}$
algorithm for \textsc{Dominating Set} in induced grid graphs of dimension $d$,
unless ETH fails.
\end{theorem}

\begin{proof}

Given an input formula $\phi$, we will replace each grid point in the
embedding of $G_\phi$ that corresponds to a variable of  $\phi$ by a
\emph{variable gadget}, each grid point that corresponds to a clause by a
\emph{clause gadget} and each wire (that is a path in the grid connecting a
variable point to a clause point) by a \emph{wire gadget}. Next we describe
these gadgets, and how they are connected to form the construction. This will be followed by a description of how the construction can be realized as a $2$- and higher-dimensional grid graph.

Our
variable gadget is a cycle of length $12$ with an ``ear'' of the same size, as
depicted in Figure~\ref{fig:domsetgadgets}. We number the vertices of the
cycle from $0$ to $11$. The idea is to represent setting the variable to true
by putting vertices with index $\equiv 1 \pmod{3}$ into the dominating set; if
the variable is false, we select vertices with index $\equiv 0 \pmod{3}$
instead. The role of the ear is to ensure that in a minimum dominating set,
one of these two scenarios is forced.

The wire gadgets are simple paths that have a length (i.e., number of edges)
that is congruent to $1$ modulo three. The clause gadget is a single vertex, see Figure~\ref{fig:domsetgadgetsfull}. A
wire that corresponds to a positive literal $x_i$ will start at vertex
$1,4$ or $7$ of the variable gadget of $x_i$, and ends at the corresponding clause vertex.
For negative literals, we start at a vertex of index $\equiv 0 \pmod{3}$, i.e., at vertex $0,3$ or $6$
instead. Note that selecting $k-1$ internal vertices from a wire of
length $3k+1$ can dominate all internal vertices of the wire. Moreover, if the literal
is true (i.e., its starting vertex in the variable cycle is in the dominating
set), then selecting every third vertex on the wire will additionally dominate the
clause vertex at the other end.

\begin{figure}[t]
\centering
\includegraphics[scale = 1]{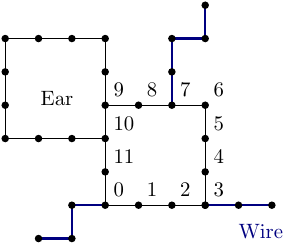}
\caption{Variable gadget for \textsc{Dominating Set}. In blue it is shown
where the wire gadgets are attached for a variable that occurs twice as a
negative literal (corresponding to wires attached to vertices 0 and 3) and
once as a positive literal (wire attached to vertex 7).}%
\label{fig:domsetgadgets}
\end{figure}

From each variable cycle, we must select at least four vertices into our
dominating set, and at least three more vertices from the ear are necessary.
Among the inner vertices of a wire of length $3k+1$, we have at least $k-1$
vertices in the dominating set.

\begin{figure}[t]
\centering
\includegraphics[scale = 1]{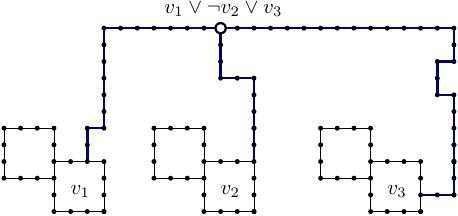}
\caption{\textsc{Dominating Set} gadgetry for the formula $(v_1 \vee \neg v_2 \vee v_3)$. The wire connected to $v_3$ has a detour to ensure that it has length $\equiv 1 \pmod{3}$.}%
\label{fig:domsetgadgetsfull}
\end{figure}

Consider a dominating-set instance corresponding to a formula
on $n$ variables with a drawing that has $w$ wires, where the wires have altogether
$\ell$ edges. The resulting grid graph has dominating set size at least $7n + \frac{\ell-4w}{3}$. It can be verified that this is attainable if and only if the formula is satisfied.
See~\cite{Berg17} for a similar, but more detailed argument.

\medskip
\noindent\emph{Two-dimensional grid graphs.}

Given an instance of
\textsc{Grid Embedded SAT}, that is, a $(3,3)$-SAT formula $\phi$ and a grid
embedded drawing $\cD$ of $G_\phi$, we need to create a grid graph which
incorporates the above gadgets. This can be done by taking a $10$-refinement
of $\cD$; this way, we can add the variable gadgets without overlap or
unwanted induced edges, and we also have space to adjust the wire length where
necessary to ensure that each wire length is $\equiv 1 \pmod{3}$. This
transformation can be done in polynomial time, and the result is an induced
grid graph drawn in an $O(n)\times O(n)$ grid, so the resulting induced grid graph has $O(n^2)$ vertices. Therefore, by Theorem~\ref{thm:GridEmbSATlower}, \textsc{Dominating
Set} has no $2^{o(\sqrt{n})}$ algorithm in induced grid graphs unless ETH
fails.

\medskip
\noindent\emph{Higher dimensional grid graphs.}

We start with a $(3,3)$-SAT formula $\phi$. We place each of the above variable
gadgets together with the first inner vertex of the connected wires in a $9\times 9 \times \dots \times 9$ small $d-1$-dimensional hypercube. These small hypercubes are then packed into a $(d-1)$-dimensional facet of a $d$-dimensional hypercube of side length
$O(n^{\frac{1}{d-1}})$. The clause gadgets along with the last inner vertices
of each wire are placed similarly in the opposing facet of the $d$-dimensional hypercube.
Applying the Strong Cube Wiring Theorem to the first and last inner vertices of the
wires that have been placed in the opposing facets, we can place each wire
inside the hypercube, by increasing the side length by a constant factor. Finally, we adjust the wires so that all of their lengths are $\equiv 1 \pmod{3}$.

The construction fits in a hypercube of side length
$O(n^{\frac{1}{d-1}})$, and the number of vertices in this induced grid graph
is $O(n^{\frac{d}{d-1}})$. Thus, a $2^{o(|V|^{1-1/d})}$ algorithm for
\textsc{Dominating Set} would translate into a
$2^{o((n^{\frac{d}{d-1}})^{1-1/d})} +\poly(n) = 2^{o(n)}$ algorithm for $(3,3)$-SAT,
contradicting ETH.
\end{proof}

\subsection{\textsc{Vertex Cover} and \textsc{Independent Set}}\label{sec:isvclower}
It is well known that \textsc{Vertex Cover} and \IS are solvable in polynomial time on bipartite graphs by using an augmenting-path algorithm. Hence, these problems are also solvable in polynomial time in $d$-dimensional grids. Consequently, we need a slightly broader graph class for this reduction. The \emph{augmented $d$-dimensional grid} for $d\ge 2$ is defined as the infinite $d$-dimensional grid graph together with the edges \index{augmented grid graph}
\[\big((x_1,x_2,\dots,x_d),(x_1+1,x_2+1,x_3,\dots,x_d)\big) \qquad \big((x_1,\dots,x_d) \in \Ints^d\big).\]
In other words, the augmented $d$-dimensional grid is obtained  from the regular $d$-dimensional grid by adding certain ``diagonals'' on 2-dimensional faces of the grid cells.

\begin{figure}[t]
\begin{center}
\includegraphics[width=0.7\textwidth]{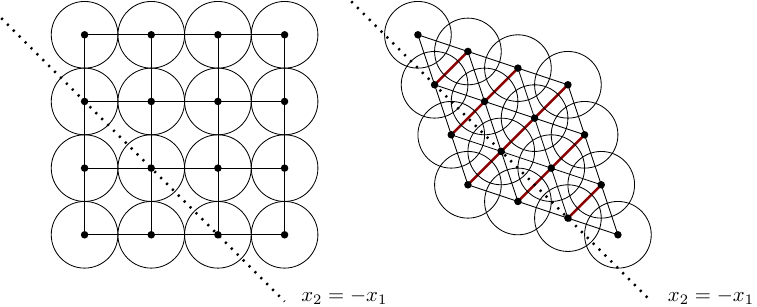}
	\caption{Realizing the augmented grid as an intersection graph of equal radius balls. The new edges introduced by the transformation are red.}\label{fig:augmentedgrid}
\end{center}
\end{figure}

The augmented $d$-dimensional grid is a unit ball graph. To see this, let $\phi: \Reals^d \rightarrow \Reals^d$ be the linear transformation 
\[
\phi(x_1,\dots,x_d) = \left(\frac{(1+\sqrt{2})x_1+(1-\sqrt{2})x_2}{2\sqrt{2}},\frac{(1-\sqrt{2})x_1+(1+\sqrt{2})x_2}{2\sqrt{2}},x_3,\dots,x_d\right),
\]
i.e., it pushes points closer to the hyperplane $x_2=-x_1$. Then the intersection graph of balls of radius $1/2$ with centers $\phi(\Ints^d)$ is the augmented grid.  See Figure~\ref{fig:augmentedgrid} for an illustration. Indeed, the distance between $\phi(x_1,...,x_d)$ and $\phi(x_1+1,x_2+1,x_3,...,x_d)$ is $1/\sqrt{2}$.

The \emph{$d$-dimensional augmented grid graphs}\index{augmented grid graph!induced \tilde} are defined as subgraphs of the augmented $d$-dimensional grid. We usually consider \emph{induced} augmented grid graphs, which are induced subgraphs of the augmented grid. Note that induced augmented grid graphs form a subclass of $d$-dimensional unit ball graphs. Instead of proving the result for unit ball graphs, we prove the following stronger statement.

\begin{theorem}
For any $d\ge 2$, \textsc{Vertex Cover} and \textsc{Independent Set} on induced augmented $d$-di\-men\-sion\-al grid graphs have no $2^{o(n^{1-1/d})}$ algorithm, unless ETH fails.
\end{theorem}

\begin{figure}[t]
\begin{center}
\includegraphics[scale=0.7]{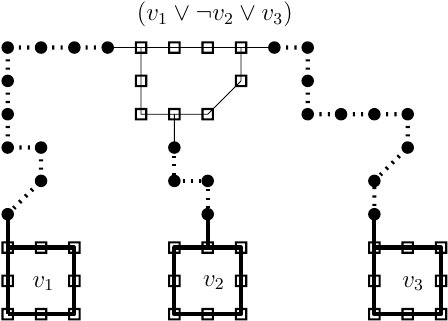}
\caption{Gadgetry for the formula $(v_1 \vee \neg v_2 \vee v_3)$.}\label{fig:vcgadget_stoc}
\end{center}
\end{figure}

\begin{proof}
The complement of an independent set is a vertex cover and vice
versa, so it is sufficient to give a reduction for \textsc{Vertex Cover}. We
will use a reduction from \textsc{Grid embedded SAT}. Let $\phi$ be a
$(3,3)$-CNF formula, and let $G_{\phi}$ be its incidence graph. Similarly to
our \textsc{Dominating Set} gadgetry, we use a cycle as variable gadget. The
literals are represented by paths of odd lengths, see
Figure~\ref{fig:vcgadget_stoc}. The variable gadget for a variable $v_i$ is a
cycle of length eight with vertices $v_i(0),\dots,v_i(7)$, where the literal
edges formerly incident to $v$ are now connected to a cycle vertex of even
index for positive and of odd index for negative literals (see
Figure~\ref{fig:vcgadget_stoc}). For a clause $c_i$ that has three literals,
the gadget is a cycle of length nine with vertices $c_i(0),\dots,c_i(8)$, and we connect the wires at  vertices $c_i(0),c_i(3),c_i(6)$. For clauses that have exactly two literals, we use an edge
as clause gadget, and connect each wire to different endpoints. We can
eliminate clauses of size one in a preprocessing step. The wire gadgets are
simple paths of odd length.

Given a satisfying assignment to $\phi$, we can create a vertex cover the following way. For true variables, we select vertices of even degree from the variable cycle, and for false variables, we select the vertices of odd degree. On the wires, we select every second vertex starting from the variable cycle; in case of a true literal, this means that we do not select the first inner vertex adjacent to the variable cycle, nor the endpoint in the clause gadget, while for false literals we do select both of these. Since this is a true assignment, at least one literal of each clause of size three is true; let $k\in {0,3,6}$ be the index at which the wire of a true literal connects to the gadget of the clause $c_i$. Then we can select vertices $c_i(k+1),c_i(k+3),c_i(k+4),c_i(k+6),c_i(k+7)$ (where the indices are interpreted modulo $9$) from the clause cycle. This covers every edge of the cycle, and also the last edges of the other wires. For a clause of size $2$, we select an endpoint of the clause segment that corresponds to a false literal, or if both literals are true, then an arbitrary endpoint.
For a construction with $\nu$ variables,
$\gamma$ clauses of three literals, $\gamma'$ clauses of two literals, and
$\kappa$ inner vertices on the literal paths, this gives a vertex cover of size
$s\eqdef 4\nu + 5\gamma + \gamma' + \kappa/2$.

Suppose now that we are given a vertex cover $S$ of size $s$. A vertex cover must contain at least four vertices of each variable cycle,
at least five vertices per clause cycle (as introduced for clauses with three literals), and at least one vertex for each
clause segment (as introduced for clauses with two literals). It is also easy to check that from a wire of
$2k+1$ edges, the vertex cover must contain at least $k$ inner vertices. It follows that $S$ is a cover where all of these inequalities are tight: each variable cycle has to contain exactly $4$ vertices of $S$, each clause cycle contains exactly $5$ vertices of $S$, each clause segment contains $1$ vertex of $S$, and among the inner vertices of each wire of length $2k+1$, there are exactly $k$ vertices from $S$.

On a variable cycle, in order to cover the cycle with four vertices, all vertices of $S$ need to have odd indices or they all need to have even indices. Let us set a variable to true if and only if the corresponding variable cycle has its even index vertices in $S$. We show that every clause is satisfied by this assignment. Consider first a clause with three literals.
We claim that for any vertex cover of size exactly five within the clause cycle,
there is at least one vertex of index divisible by three that is not in $S$. Suppose the contrary: the clause cycle of $c_i$ satisfies $c_i(0),c_i(3),c_i(6) \in S$. Then $S$ is disjoint from at least one of the sets $\{c_i(1),c_i(2)\},\{c_i(4),c_i(5)\},\{c_i(7),c_i(8)\}$, thus $S$ fails to cover at least one edge among $c_i(1)c_i(2),c_i(4)c_i(5),c_i(7)c_i(8)$, which is a contradiction.
It follows that the last edge of at least one of the wires connecting to the clause cycle is covered by the last inner vertex of the wire. By the size bound, we know that this wire has every second inner vertex in $S$, so it follows that the starting vertex of the wire in the variable gadget has to be in $S$, which means that the corresponding literal is true. A somewhat simpler argument shows that clauses of size $2$ are also satisfied by this assignment. 

Therefore, there is a vertex cover of size $s$ if and only if the original formula is
satisfiable.

Next, we need to insert these gadgets into a refined version of either the
$2$-di\-men\-sion\-al grid embedding or the cube wiring. We regard this
refined grid as a subgraph of the augmented grid. Using the ``diagonals'' of
the augmented grid, the odd-length clause cycles can be realized. We can also
enforce the parity condition on the wires by incorporating some diagonals; we
introduce small local detours on the wires that have even length after the
refinement to make their length odd. (Figure~\ref{fig:vcgadget_stoc} has two
wires with local parity adjustments.)
\end{proof}

\subsection{\textsc{Connected Vertex Cover}}\label{sec:cvclower}

In this section we prove the following theorem.

\begin{theorem}\ \\[-1.5em]
\begin{itemize}
\item \textsc{Connected Vertex Cover} has no $2^{o(\sqrt{n})}$ algorithm
in $2$-dimensional grid graphs\footnote{This refers to subgraphs of the grid that are not necessarily induced.} or in unit disk graphs, unless ETH fails.
\item Let $d\geq 3$. Then \textsc{Connected Vertex Cover} has
no $2^{o(n^{1-1/d})}$ algorithm in induced $d$-dimensional grid graphs, unless ETH fails.
\end{itemize}
\end{theorem}

We apply ideas from a reduction by Garey and Johnson~\cite{GareyJ77} to make
our gadgetry for \textsc{Connected Vertex Cover} from our original
\textsc{Vertex Cover} gadgets. The key step of their reduction
converts a planar graph of maximum degree three into a planar graph of maximum degree four in such a way that a vertex cover of the original
graph corresponds to a connected vertex cover of the constructed graph. See Lemma
2 in their paper, which effectively adds a \emph{skeleton} to the graph. We\index{skeleton}
illustrate this on a five-vertex augmented grid graph instead of on our actual
construction; see the first step of Figure~\ref{fig:cvc_example}.

By adding the skeleton to our \textsc{Vertex Cover} construction, we get a
planar graph of maximum degree four, which has a connected vertex cover of size $k$ if and only if the original graph has a vertex cover of some specific size $k'$. Starting from our previous augmented grid
embedding, the edges of this planar graph can be drawn as grid paths in a
refinement of the \textsc{Vertex Cover} drawing. (Note that we avoid the
diagonals with these paths, and only use grid edges in the drawing.) We call
this grid drawing $\cD_\phi$. We use the following simple trick~\cite{EscoffierGM10} to make an equivalent instance that is a grid graph.

\begin{observation}\label{obs:getequiv}
Let $e=uv$ be an edge of a graph $G$ that has at least two edges. Let $G'$ be the graph that we get if we
subdivide $e$ using a vertex $w_e$, and add a leaf $w'_e$ that is connected to $w_e$ (i.e., $V(G') = V(G) \cup
\{w_e,w'_e\},\; E(G') = \big(E(G) \setminus \{uv\}\big) \cup \{uw_e,w_e v,w_e
w'_e\}$). Then $G$ has a connected vertex cover of size $k$ if and only if
$G'$ has a connected vertex cover of size $k+1$.
\end{observation}

\begin{figure}[t!]
\begin{center}
\includegraphics[width=0.8\textwidth]{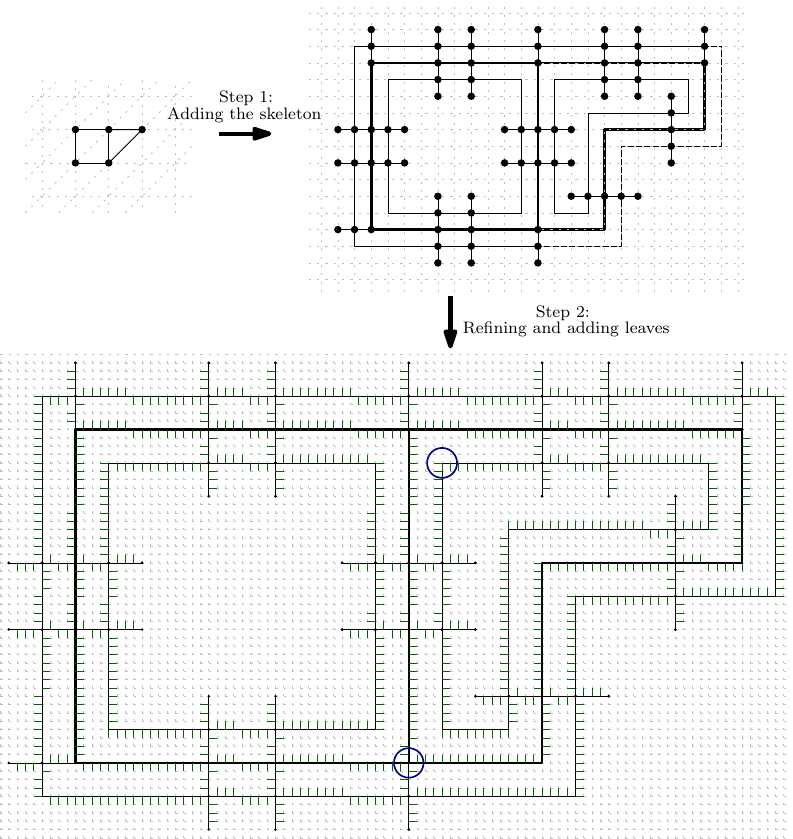}
\caption{Two transformations starting from a small augmented grid graph,
resulting in a (non-augmented) grid graph. The leaves can be added without
conflicts, even around degree four vertices and corners (in blue circles).}%
\label{fig:cvc_example}
\end{center}
\end{figure}

We refine $\cD_\phi$ by a factor of four; this way each old edge becomes a
grid path of length at least four. We subdivide each edge by adding all the
grid points that lie on its grid path as vertices, and we add leaves to all of
these new vertices. This corresponds to applying
Observation~\ref{obs:getequiv} multiple times, therefore we get an equivalent
instance. Note that we need to show that these leaves can be added into the current grid drawing without trying to assign the same grid point to two different leaves. (We call such double assignments \emph{conflicts}.) 
Consider first the neighborhoods of
\emph{interesting vertices}, that is, in neighborhoods of vertices of degree
four and corners. Notice that we can choose the leaves in a manner (see Figure~\ref{fig:cvc_example}) that avoids the conflicts. For all other vertices, the leaves cannot introduce any
conflict, as for any pair of leaves $a,b$, we have that the neighbor of $a$ and the neighbor of $b$ are already assigned to grid points that have no shared neighbors. Since we only used constantly many refinements, the
resulting grid graph is drawn in an $O(n)\times O(n)$ grid.

We have shown how to modify the 2-dimensional lower-bound construction for \textsc{Vertex Cover} to obtain an instance of \textsc{Connected Vertex Cover}, such that the original instance of \textsc{Vertex Cover} has a solution of size $k$ if and only if the resulting instance of \textsc{Connected Vertex Cover} has a solution of some specific size $k''$. Consequently, there is no $2^{o(\sqrt{n})}$ algorithm for \textsc{Connected Vertex Cover} in $2$-dimensional grid graphs unless ETH fails.

\begin{remark}\label{rem:conn_ext}
Let $V$ be the vertex set of the starting {\sc Vertex Cover} construction, and let $W$ and $L$ denote the newly introduced subdividing and leaf vertices respectively. Let $G$ be the constructed graph itself. Note that $W$ is a minimum (not connected) vertex cover of $G$, since each edge of the original vertex cover construction has been subdivided at least once. An easy argument shows that there is a minimum connected vertex cover which contains $W$, therefore a minimum connected vertex cover is essentially the smallest set $S\supseteq W$ that induces a connected subgraph of $G$.
\end{remark}

Note that we can realize our construction above as a unit disk graph. The disk
centers are the same as in the grid graph, but we use disks of radius $1/2$.
Moreover, we shift the disk centers corresponding to leaf vertices by $1/3$ or
$1/4$ towards the neighboring disk's center, with the constraint that leaves
added to neighboring vertices get a different shift.

\smallskip
When $d\geq 3$, we can do the same modifications. By being careful
with adding the leaves we can even get a $d$-dimensional \emph{induced} grid
graph. Indeed, it is easy to avoid conflicts and unwanted induced edges between leaves that are
attached to neighbors of interesting vertices. (The most challenging case is
vertices of degree four, but these have the property that the four neighboring
edges lie in the same 2-flat, so by placing the leaves outside this 2-flat we
can avoid conflicts between them.) As for the straight paths of length at least four
connecting these vertices, the leaves can be adjusted
on the paths so that the leaves attached to the first and last inner vertex are
pointing in the desired direction, as required by the interesting vertex at the
endpoint.

\subsection{Some easy consequences}\label{sec:brieflowers}

This section sketches further ETH-based lower bounds of the form $2^{\Omega(n^{1-1/d})}$ for several problems.

\paragraph*{\textsc{Steiner Tree}}

We apply a $2$-refinement to our connected vertex cover construction from Section~\ref{sec:cvclower}. We then
subdivide every edge with the new grid point in the middle, and define the set
of terminals to be these new vertices. The non-terminal vertices of a Steiner tree in this
graph is a connected vertex cover in the original  graph and the other way around.
Notice that due to the refinement, the resulting graph is an induced grid
graph even in the $2$-dimensional case.

\paragraph*{\textsc{Connected Dominating Set}}

We use a classical reduction by Clark~\etal~\cite{ClarkCJ90} from \textsc{Planar
Connected Vertex Cover} to \textsc{Grid Connected Dominating Set} (see Theorem 5.1
in~\cite{ClarkCJ90}). We apply this reduction to our \textsc{Vertex Cover}
construction from Section~\ref{sec:isvclower}. We get an induced grid graph embedded in an $O(n)\times
O(n)$ grid. We can divide the construction into constant size
variable and clause gadgets, and wire gadgets of size proportional to their
length, and use these gadgets for the higher dimensional reduction, similar to what we did at the end of Section~\ref{sec:dslower}.

\paragraph*{\textsc{Feedback Vertex Set}}

We observe that subdividing an edge in a \textsc{Feedback Vertex Set}
instance leads to an equivalent instance. Take our \textsc{Vertex Cover}
construction from Section~\ref{sec:isvclower}, and add a triangle to each
edge, that is, for each edge $uv$ we add a new unique vertex $w_{uv}$, and the
edges $uw_{uv}$ and $uw_{uv}$. These triangles ensure that at least one
endpoint of the original edge has to be in the feedback vertex set. Moreover,
the original graph has a vertex cover of size $k$ if and only if the new graph
has a feedback vertex set of size $k$. This results in a planar graph of
degree at most $6$. Our wires become triangle chains (of degree at most $4$),
and using subdivisions we can realize this wire gadget as an induced grid
graph.

For vertices of degree more than $4$, we use the degree reduction gadget by
Speckenmeyer~\cite{Speckenmeyer83}, which gives us constant size planar
graphs that can be put in place of vertices of degree $5$ and $6$. These planar graphs can
be drawn in an $O(1)\times O(1)$ grid, which can be turned into an induced
grid graph of constant size using subdivisions. We introduce a refinement so
that we can insert these gadgets as necessary.

\subsection{\textsc{Connected Feedback Vertex Set}}

The goal of this section is to prove the following.

\begin{theorem}\label{thm:cfvslower}\ \\[-1.5em]
\begin{enumerate}
\item[(i)] \textsc{Connected Feedback Vertex Set} has no $2^{o(\sqrt{n})}$ algorithm
in unit disk graphs, unless ETH fails.
\item[(ii)] Let $d\geq 3$. Then \textsc{Connected Feedback Vertex Set} has
no $2^{o(n^{1-1/d})}$ algorithm in induced $d$-dimensional grid graphs, unless ETH fails.
\end{enumerate}
\end{theorem}

\begin{proof}
We use our {\sc Connected Vertex Cover} construction as a starting point, but replace each leaf added in Observation~\ref{obs:getequiv} with a \emph{leaf cycle}: for a vertex $w$ with leaf $\ell$ we add another vertex $\ell'$ and connect $w\ell$ and $\ell\ell'$ to form a cycle of length three. There exists a minimum connected feedback vertex set $S$ which contains all the vertices $w$ that have such a leaf cycle attached. Note that the set $W$ of vertices with attached leaf cycles is by itself a feedback vertex set, so similarly to Remark~\ref{rem:conn_ext}, a minimum connected feedback vertex set $S$ is just the smallest set containing $W$ that induces a connected subgraph.

To prove (i), note that the construction without leaf cycles is a grid graph, and it has a natural unit disk representation. We can extend this representation at each vertex $w\in W$ by adding two small perturbations of the disk of $w$, making sure that neighboring vertices $w,w'\in W$ get disjoint disks. This results in a unit disk representation of our construction.

To prove (ii), we use the same global structure, starting from our $d$-dimensional {\sc Connected Vertex Cover} construction, but we need to introduce some further changes.
We  replace our triangular leaf cycles with cycles of length $4$, that is, we subdivide each edge $\ell\ell'$ with a new vertex; let $G^d$ denote the resulting graph. We use cube wiring to realize $G^d$ as a non-induced grid graph in $\Reals^3$ and in higher dimensions. Since our original connected vertex cover construction had maximum degree 4, the graph $G^d$ has maximum degree 5, which is only attained on vertices that have a leaf cycle attached. Let us call an induced path in this construction which has a leaf cycle on all of its internal vertices a \emph{leaf path}. We argue that in $\Reals^3$, the graph $G^3$ can be tweaked to get an induced grid graph; the same tweaks can be used in higher dimensions as well.

\begin{figure}[t]
\begin{center}
\includegraphics[width=0.6\textwidth,trim={0cm 0cm 0cm 2cm},clip]{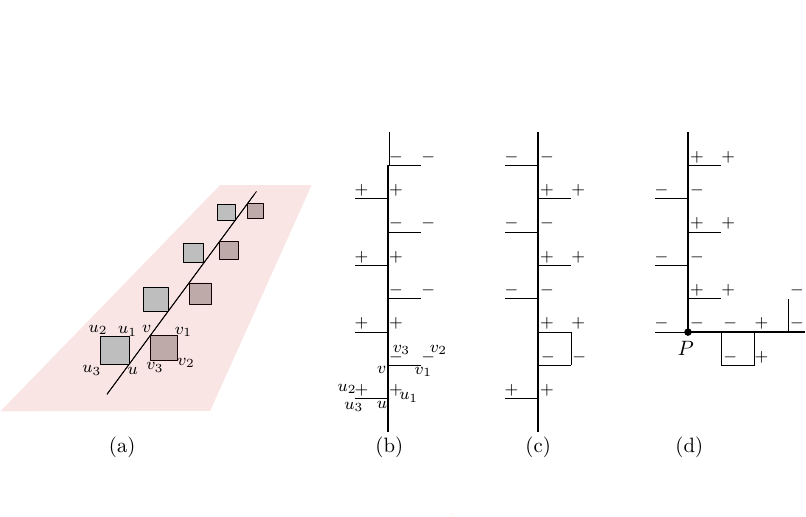}
\caption{(a) and (b): A path with leaf cycles. The cycles on the left are above,  and the cycles on the right are below the plane $z=0$. (c): A tweak along a path with leaf cycles. (d):Introducing a turn in a leaf path.}\label{fig:cfvs_3d_paths}
\end{center}
\end{figure}

\begin{figure}[h]
\begin{center}
\includegraphics[width=\textwidth]{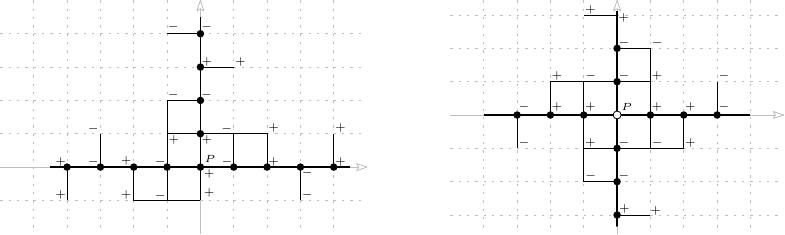}
\caption{Left: Gadget for replacing degree five skeleton vertices. Right: Gadget for replacing degree five normal vertices.}\label{fig:cfvs_3d_gadgets}
\end{center}
\end{figure}

We can avoid unwanted induced edges along leaf paths by putting the leaf cycles along each path on alternating sides (see Figure~\ref{fig:cfvs_3d_paths}a). In order to visualize our gadgets in figures, we usually just draw the intersection of the gadget with a plane, which we can assume to be the plane $z=0$. If the grid point directly above some vertex in this plane is in the gadget, e.g., we have $(x,y,1)$ in the gadget, then we put the sign $+$ near the point $(x,y,0)$ . Similarly, the sign $-$ near $(x,y,0)$ denotes that $(x,y,-1)$ is included in the gadget. For example, the depiction in Figure~\ref{fig:cfvs_3d_paths}a is equivalent to Figure~\ref{fig:cfvs_3d_paths}b.

If we cannot place the leaf cycles alternatingly along a path for some reason, then we can introduce a tweak as shown in Figure~\ref{fig:cfvs_3d_paths}c. Note that the tweak does not change the structure or size of a minimum feedback vertex set. Furthermore, we can also introduce turns on leaf paths: we show a turn at a vertex $P$ in Figure~\ref{fig:cfvs_3d_paths}d.

Finally, we need gadgets to deal with potential unwanted induced edges around vertices of degree $4$ without leaf cycles, and around vertices of degree $5$. 

Due to the construction, a vertex of degree $5$ has a leaf cycle and is the common endpoint of three leaf paths. We replace such vertices with the gadget on the left in Figure~\ref{fig:cfvs_3d_gadgets}. Notice that the gadget does not contain new vertices, only some extra edges. All the marked vertices (the vertices that are in $W$) can be assigned to vertex-disjoint cycles. For the central vertex $P=(0,0,0)$ and its selected neighbors, the cycles are
\[
\begin{matrix}
(0,0,0),&(0,0,1),&(0,-1,1,),&(0,-1,0);\\
(1,0,0),&(1,0,-1),&(1,1,-1),&(1,1,0);\\
(0,1,0),&(0,1,1),&(-1,1,1),&(-1,1,0);\\
(-1,0,0),&(-1,0,-1),&(-1,-1,-1),&(-1,-1,0).
\end{matrix}
\]
The gadget without the selected vertices is cycle-free. Therefore, there is a minimum connected feedback vertex set containing the selected vertices, i.e., the gadget has the same role as in the original construction, and the feedback vertex set itself remains unchanged.

For vertices of degree $4$ without a leaf cycle we introduce the gadget on the right of Figure~\ref{fig:cfvs_3d_gadgets}. Again, the marked vertices (the vertices from $W$) can be assigned to vertex disjoint cycles, and their removal gives a cycle-free graph: from the gadget, we will be left with the vertex $P$ and four induced paths of length $6$.

Using a further refinement, the above gadgetry can be integrated into the construction to create an induced $d$-dimensional grid graph that is at most constant times larger than the original graph. Therefore a $2^{o(n^{1-1/d})}$ algorithm for \textsc{Connected Feedback Vertex Set} would give a $2^{o(n^{1-1/d})}$ algorithm for \textsc{Feedback Vertex Set}, which would violate ETH, thus finishing the proof.
\end{proof}

\subsection{\textsc{Hamiltonian Cycle}, \textsc{Hamiltonian Path}, and \textsc{Euclidean TSP}}

We prove the following theorem first.

\begin{theorem}
There is no $2^{o(n^{1-1/d})}$ algorithm for 
\textsc{Hamiltonian Cycle} or \textsc{Hamiltonian Path} in induced $d$-dimensional grid graphs, unless ETH fails.
\end{theorem}

\begin{proof}
We can essentially use the construction by Itai~\etal~\cite{ItaiPS82} for \textsc{Hamiltonian Cycle} in grid graphs, but in order to get a tight bound, we need to be somewhat careful with how we handle the underlying grid embedding. The writeup assumes that the reader is familiar with~\cite{ItaiPS82} and~\cite{Plesnik79}.
Our proof is a reduction from $(3,3)$-SAT that uses several steps. Similarly to the other problems, we start with the planar case.

Given a $(3,3)$-CNF formula $\phi$, we start by drawing its incidence graph in the grid the following way. We place the variable vertices horizontally on the top, and the clauses vertically on the left of the figure, each edge is drawn in a shape ``$\lrcorner$''. Next, we apply the construction of Plesn\'ik~\cite{Plesnik79} for the \NP-completeness of \textsc{Directed Hamiltonian Cycle} in planar digraphs where the sum of in- and outdegrees of each vertex is $3$ to this specific drawing, see Figure~\ref{fig:plesnikgrid}. (Note that the gadgetry is similar to, but slightly different from the one given by Garey, Johnson, and Tarjan~\cite{GareyJT76}.) In this gadgetry, each variable and its negation is assigned a pair of parallel arcs, and the truth value is determined by the Hamiltonian cycle (the arc contained in the Hamiltonian cycle is exactly one of the two arcs). These parallel arcs are connected by XOR-gadgets, which are essentially four arcs, alternatingly oriented. The clauses are represented by three or two pairs of parallel arcs, depending on the number of literals inside. These arcs are attached to triple-OR and normal OR gadgets.
The opposing arc of each literal is connected to an arc of the corresponding variable with a XOR gadget, that is, if variable $x$ occurs as a positive literal in clause $c$, then we enforce the condition that either $x$ is true and the negation of the literal in $c$ is false, or $x$ is false and the negation of the literal is true. Such XOR-gadgets can also cross each other using a crossing gadget. It is easy to see that for a $(3,3)$-CNF formula $\phi$ on $n$ variables, the obtained planar digraph $G_1$ has size $O(n^2)$, and moreover that it is drawn in an $O(n)\times O(n)$ grid.

\begin{figure}[t]
\begin{center}
\includegraphics[width=\textwidth]{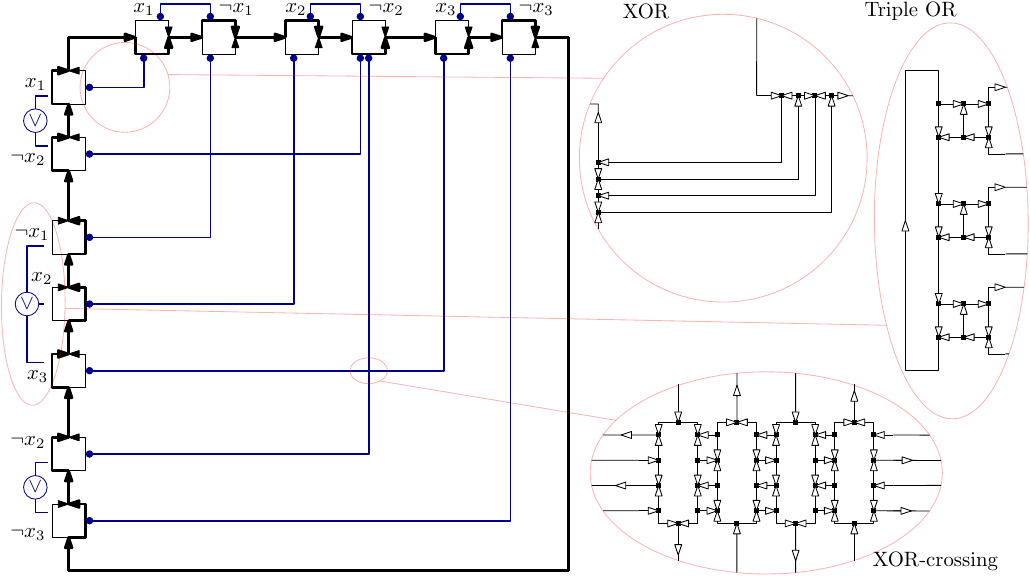}
\caption{The construction by Plesn\'ik for the formula $(x_1\vee \neg x_2)\wedge(\neg x_1\vee x_2 \vee x_3)\wedge(\neg x_2\vee \neg x_3)$, drawn in an $O(n)\times O(n)$ grid.}\label{fig:plesnikgrid}
\end{center}
\end{figure}

The next reduction step is to \textsc{Hamiltonian Cycle} in planar undirected bipartite graphs
(see also~\cite{ItaiPS82}); one can just replace each vertex $v$ of $G_1$ with two vertices, $v_{in}$ and $v_{out}$, connected by an edge, and for each arc $uv$ of $G_1$, we add the edge $u_{out}v_{in}$ to the new graph $G_2$. Note that by introducing a $2$-refinement in the drawing of $G_1$, we can add the new vertices and change the edges accordingly, therefore we get a drawing of $G_2$ in an $O(n)\times O(n)$ grid. Using this graph, we can follow the proof by Itai, Papadimitriou, and Szwarcfiter~\cite{ItaiPS82} from this point onwards: we can make the above drawing of $G_2$ into a \emph{parity preserving embedding} in a 3-refinement. This is a grid drawing where the left and right side of the bipartite graph are mapped to even and odd grid points respectively. Next, they take a 9-refinement of the underlying grid, and substitute vertices with the nine vertices covered by a $2\times 2$ grid square. Each edge becomes a \emph{tentacle}, connecting two such grid squares. A tentacle is a width $2$ grid path (see Figure~\ref{fig:tentacle}), which allows one to model the Hamiltonian cycle passing the edge by ``snaking'' through it, or to do a long detour, which happens for tentacles that correspond to edges that are not in the original Hamiltonian cycle in $G_2$. The final graph we get is an induced grid graph, that is now guaranteed to fit in an $O(n)\times O(n)$ grid. This concludes the proof of the lower bound in $2$ dimensions.

\begin{figure}[t]
\begin{center}
\includegraphics[height=4cm]{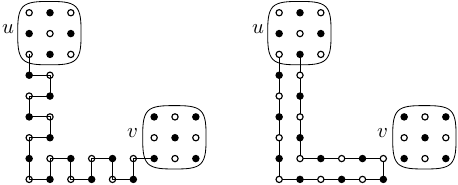}
\caption{A Hamiltonian cycle passing from vertex $u$ to $v$ through a tentacle (left), and making a detour on the same tentacle (right).}\label{fig:tentacle}
\end{center}
\end{figure}

\begin{figure}
\begin{center}
\includegraphics[height=6cm]{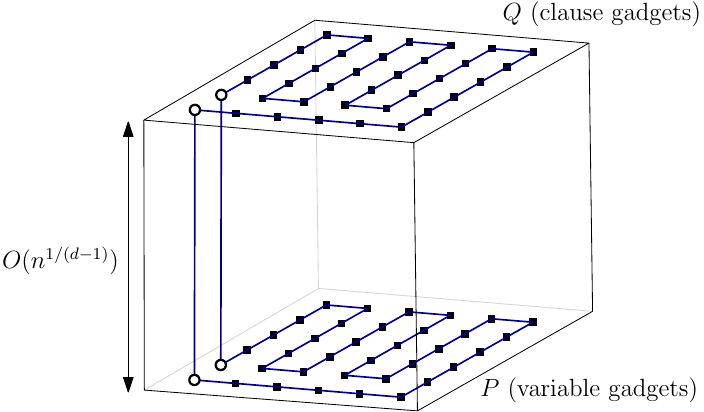}
\caption{Adding a strip through all the variable and clause gadgets. This a schematic picture; the points represent gadgets, and the blue segments correspond to tentacles.}\label{fig:3dhamcycle}
\end{center}
\end{figure}

For higher dimensions, we can reuse the variable, clause and wire (XOR) gadgets we gained in the $2$-dimensional construction. Notice that the XOR-crossing gadgets are not necessary. The one additional thing to take care of is that we need to place the variable gadgets and the clause gadgets along a cycle, that is, we need to run a tentacle through these gadgets; essentially, we need to run this strip through our point sets $P$ and $Q$ in the cube wiring. To this end, one just needs to take a Hamiltonian path on the variable gadgets in the bottom facet, and connect its ends to the ends of the Hamiltonian path drawn on the clause gadgets in the top facet. We illustrate the approach in three dimensions in Figure~\ref{fig:3dhamcycle}.

To show the same bound for \textsc{Hamiltonian Path}, observe that there are edges in our \textsc{Hamiltonian Cycle} construction that are contained in all Hamiltonian cycles. Such an edge can be drawn as a simple path in the grid, instead of a tentacle as for other edges. By removing an inner vertex $v$ of such a path, the neighbors of $v$ will have degree $1$, therefore a Hamiltonian path is forced to have these vertices as endpoints. Consequently, we have gained an equivalent instance.
\end{proof}

The following is an important corollary.

\begin{corollary}
There is no $2^{o(n^{1-1/d})}$ algorithm for 
\etsp in $\Reals^d$, unless ETH fails.
\end{corollary}
%--------------------------------------------------------------------------------------------
\begin{proof} 
Let $G$ be the induced grid graph constructed above for Hamiltonian Cycle, and let $P\subseteq \Reals^d$ the set of grid points realizing its vertex set.
We claim that $G$ has a Hamiltonian cycle if and only if $P$ has a TSP tour of length $n$. 
Indeed, a Hamiltonian cycle in $G$ induces a tour of length $n$ 
on~$P$, since edges of the induced grid graph $G$ correspond to pairs of vertices at distance $1$. 
Moreover, if there is a tour of length $n$ on~$P$, then 
there must be a Hamiltonian cycle in $G$, since a tour of length $n$ cannot
afford to use edges of length greater than~1 (as there are no edges of length smaller than~1, and pairs of points at distance exactly one are always connected in $G$).
\end{proof}

%% file: conclusion.tex
\section{Conclusion}

We have presented an algorithmic and lower bound framework for obtaining $2^{\Theta(n^{1-1/d})}$ algorithms and matching conditional lower bounds for several problems in geometric intersection graphs. We find the following questions intriguing:
\smallskip
\begin{itemize}
\item Is it possible to obtain clique decompositions without geometric information? Alternatively, how hard is it color the complement of a small diameter geometric intersection graph of fat objects? If we could approximate the size of a minimum clique decomposition in a greedy partition class within a constant factor, then all our algorithms (except those for {\sc Hamiltonian Cycle} and {\sc Hamiltonian Path}) would become robust.
\item Many of our applications require the low degree property (i.e., the fact that $G_\cP$ has bounded degree). Is the low degree property really essential for these applications? Would having low average degree be sufficient?
\item Is it possible to modify the framework to work without the similar size assumption? In case of \textsc{Dominating Set}, it is already known that the similar size assumption is necessary~\cite{Bergk20}.
\end{itemize}
\smallskip
Finally, it would be interesting to explore the potential consequences of this framework for parameterized and approximation algorithms. Work in this direction has already begun~\cite{FominLP0Z20}.